\documentclass[11pt,envcountsame,envcountsect]{llncs}
\usepackage[utf8]{inputenc}
\usepackage{amsmath,amssymb}
\usepackage{fullpage}
\usepackage[algoruled,vlined,english]{algorithm2e}
\usepackage{xspace}
\usepackage{tikz}
\tikzset{defnode/.style={draw,darkgray,circle,fill,inner sep=1.2pt}}
\tikzset{loopbelow/.style={->,loop,in=235,out=305,looseness=14}}
\tikzset{loopabove/.style={->,loop,in=55,out=125,looseness=14}}
\tikzset{loopleft/.style={->,loop,in=145,out=215,looseness=14}}
\tikzset{loopright/.style={->,loop,in=-35,out=35,looseness=14}}
\tikzset{loopne/.style={->,loop,in=10,out=80,looseness=14}}
\tikzset{loopnw/.style={->,loop,in=100,out=170,looseness=14}}
\tikzset{loopsw/.style={->,loop,in=190,out=260,looseness=14}}
\tikzset{loopse/.style={->,loop,in=280,out=350,looseness=14}}
\newcommand{\ie}{i.e.,\xspace}
\newcommand{\mk}{\ensuremath{{\mathcal{M}}}}
\newcommand{\N}{\ensuremath{\mathbb N}\xspace}
\newcommand{\NNNN}{\mathcal{P}(\N^3)\xspace}

\usepackage[color=green!10!white,linecolor=green!50,textwidth=2.5cm,textsize=small]{todonotes}
\begin{document}
\title{Revisiting the Role of Coverings in Anonymous Networks: Spanning Tree Construction and Topology Recognition}
\author{Arnaud Casteigts, Yves M\'etivier, and John Michael Robson}
\institute{LaBRI, Universit\'e de Bordeaux, CNRS, Bordeaux INP, France\\
  \{acasteig, metivier, robson\}@labri.fr }
\date{} 
\pagestyle{plain}
\maketitle

\begin{abstract}
  This paper revisits two classical distributed problems in anonymous networks, namely spanning tree construction and topology recognition, from the point of view of graph covering theory. For both problems, we characterize necessary and sufficient conditions on the communication graph in terms of directed symmetric coverings. These characterizations answer a long-standing open question posed by Yamashita and Kameda~\cite{YKsolvable}, and shed new light on the connection between coverings and the concepts of views and quotient graphs developed by the same authors.
  
  Characterizing conditions in terms of coverings is significant because it connects the field with a vast body of classical literature in graph theory and algebraic topology. In particular, it gives access to powerful tools such as Reidemeister's theorem and Mazurkiewicz's algorithm. Combined together, these tools allow us to present elegant proofs of otherwise intricate results, and their constructive nature makes them effectively usable in the algorithms.

This paper also gives us the opportunity to present the field of covering theory in a pedagogical way, with a focus on the two aforementioned tools, whose potential impact goes beyond the specific problems considered in this work.
\end{abstract}

\setcounter{page}{1}
\setcounter{footnote}{0}

\section{Introduction}

An algorithm is distributed if its execution takes place on several entities of a system that communicate and coordinate their actions to achieve a common goal. Such a paradigm covers a wide range of topics and settings. In some settings, the \emph{feasibility} of the distributed task is certain, and the main question is the cost at which the task can be realized, \ie its complexity in terms of time and communication. When feasibility is not certain, the natural question is to characterize the conditions under which the task is feasible, and understand what are the main obstructions.

Many sources of impossibility in distributed computing stem from symmetries in the system, which prevent the entities (also called nodes or processors) from making proper decisions. A classical example is the impossibility of consensus in fully asynchronous systems in presence of a single crash~\cite{FLP85}. When the nodes of the system communicate through a network, the structure of the network itself is a possible source of symmetry. In some settings, these symmetries can be broken using randomness and/or unique identifiers. Without such features, \ie when one considers deterministic algorithms in anonymous networks, it is crucial to understand the kinds of symmetry that prevent a problem from being solved.

In a seminal work, Angluin~\cite{Angluin} initiates the study of impossibility results in anonymous networks. In particular, she shows that no deterministic algorithm exists that solves leader election and related problems in presence of certain symmetries. Here, the symmetries are not global ones, but a form of \emph{local similarity} which prevents the nodes from distinguishing between several execution scenarios. The key concept in these characterizations is that of a \emph{covering} relation, \ie a homomorphism $\varphi$ from the actual communication graph $G$ to a (possibly smaller) graph $H$, which preserves neighborhood. Angluin's \emph{lifting lemma} establishes that, if such a relation is proper, then given the execution of an algorithm in $H$, there exists a possible execution of the same algorithm in $G$ which acts indistinguishably on those vertices of $G$ that $\varphi$ sends to a same vertex of $H$. For instance, if a leader is elected in $H$, then several leaders may be elected in $G$, which contradicts the existence of an election algorithm for $G$.

A few years later, Yamashita and Kameda started a systematic study of four representative problems in anonymous networks through a series of articles~\cite{YK88,YKsolvable,YKelection}. These problems are {\em leader election}, {\em edge election}, {\em spanning tree construction}, and {\em topology recognition}. In particular, they obtain in~\cite{YKsolvable} a rich set of characterizations pertaining to different types of symmetries. The key ingredient in these works is that of a \emph{view}. The view of a vertex $v$ is an infinite tree rooted at $v$, which unfolds recursively the structural information available to $v$ (through its neighbors, the neighbors of its neighbors, \textit{etc.}).
The \emph{symmetricity} of a graph in general is then defined as the maximum multiplicity among similar views (over all port numberings). Finally, they define a concept of \emph{quotient graph} in which every dissimilar view is represented by a single vertex, and which plays a similar role as a minimum graph for the covering relation.

While both families of tools---those based on views and those based on coverings---serve the same essential purpose, the way they relate to each other is not straightforward.
One source of confusion is that the computational models considered in~\cite{Angluin} and~\cite{YKsolvable} are different. Angluin considers a coarse-grain pairwise interaction model in which the local symmetry is broken upon interaction between two neighboring nodes (in the same spirit as the more recent population protocols~\cite{population}). In contrast, Yamashita and Kameda consider an asynchronous message passing model with neither atomicity nor local symmetry-breaking facilities. (Both use locally-unique port numbering.) For one of the problems considered, \emph{topology recognition}, Yamashita and Kameda nonetheless express a sufficient condition in terms of Angluin's coverings, warning the reader that the condition is not necessary.

In \cite{BVanonymous},  Boldi and Vigna establish a fundamental link between views and directed graph coverings, considering the notion of ``fibrations'' (a homomorphism between two directed graphs which preserves the outgoing arcs): 
two vertices have the same view if and only if they lie in the same
fibre with respect to a minimal fibration. They also
give algebraic characterizations of computability in anonymous networks, where fibrations play the central role (see~\cite{BVelection} for the particular case of election and~\cite{BVfibrations} for a general overview). A key contribution in these works is that asynchronous message passing can indeed be studied through the lenses of coverings, but on the condition that the communication graph be handled as a \emph{symmetric directed graph}. This subtle, and somewhat counterintuitive aspect, of using directed graphs to deal with undirected networks, is what enables (among others) a proper treatement of loops in the minimum graph, which posed some problems in the case of quotient graphs. In particular, the homomorphism that sends a graph to its quotient graph in~\cite{YKsolvable} does not always induce a bijection among local port numbers due to some of these loops.

\subsection{Contributions}

The contributions of this paper are manifold. One motivation is to revisit in a pedagogical way a number of tools and concepts from the theory of graph coverings. They include, in particular, several versions of graph coverings, in relation to different models of computation, whose correspondence was established in~\cite{Cthese,CMasynj}. In particular, we focus on \emph{symmetric directed coverings}, in correspondence to asynchronous message passing. We review two landmark tools in covering theory. The first, Reidemeister's Theorem~\cite{Reidemeister}, gives a powerful and \textit{constructive} way to relate a graph to all its coverings of a certain size (the original theorem is from early algebraic topology and seems to have been imported into the algorithmic field by Bodlaender~\cite{Bod}). The second, also constructive, is Mazurkiewicz's algorithm~\cite{MazurEnum}, a distributed algorithm that computes a graph of which the actual communication graph is a covering. This computation uses a polynomial number of messages (of polynomial size) and does not rely on views. In particular, we focus on the version of Mazurkiewicz's algorithm which computes symmetric directed coverings and operates in asynchronous message passing. This particular version was defined in~\cite{CMasynj} and used to characterize the fact that, in this communication model, the election problem is solvable in a graph $G$ if and only if $Dir(G)$ is minimal for the symmetric covering relation.

The second contribution is the formulation of necessary and sufficient conditions, in terms of coverings, for the problems of \emph{spanning tree construction} and \emph{topology recognition}, when the number of nodes $n$ is known. These results convey in part notions and ideas similar to the results in~\cite{YKsolvable}; however, the formulation of these conditions in the langage of coverings is significant for specific reasons (detailed below) and it was a wish expressed by Yamashita and Kameda (Section 7.3 of~\cite{YKsolvable}). The exact formulation of these characterizations relies on the symmetric directed version of a graph $G$, $Dir(G)$, obtained by replacing every edge of $G$ by two opposite arcs between the same pair of vertices.
We prove that the graphs for which the spanning tree computation
problem is solvable are the ones such that either $Dir(G)$ is minimal for the symmetric 
covering relation,
or $Dir(G)$ is a $2$-sheeted symmetric 
covering of a symmetric digraph having at least one
loop  and $Dir(G)$ is not a $q$-sheeted symmetric covering of a symmetric digraph with $q>2$ (Theorem~\ref{un}). (A $q$-sheeted covering of a graph $D$ is a covering where every vertex of $D$ has $q$ inverse images.)
The graphs for which the topology recognition
problem is solvable are the graphs $G$ such that $Dir(G)$ is the unique (connected) symmetric covering of its size with respect to all symmetric 
digraphs $D$ such that $Dir(G)$ is a symmetric covering of $D$ (Theorem~\ref{deux}).

Beyond intrinsic interest, these characterizations are significant for their connection with covering machinery. In particular, Mazurkiewicz's algorithm possesses features that are subtly different from algorithms computing the views and the quotient graph. To see this, some background is required. 
Norris showed that the views can be truncated to depth $n-1$~\cite{Norris}, and Tani presented a compression technique that makes it possible to compute the views using polynomially many messages of polynomial size~\cite{compressed}. Then, the quotient graph can be inferred from the views in polynomial time~\cite{YKsolvable}. This workflow is deterministic in the sense that the views capture exactly the same information as the quotient graph, and so, independently from the contingency of an asynchronous execution. In contrast, Mazurkiewicz's is susceptible of computing any graph $H$ of which the actual communication graph $G$ is a covering (i.e., not necessarily the smallest among them), depending on the particular ordering of events in the execution. This should be seen as an advantage only, as the actual minimum graph (analog of quotient graph in terms of coverings) can still be inferred from $H$ and its port numbers in polynomial time~\cite{paulusma}, and on the other hand, $H$ may convey more information about $G$ than the minimum graph does. In particular, if we consider the characterization of feasible cases for topology recognition, some executions may yield a graph $H$ such that $G$ is uniquely determined from $H$, although the actual minimum graph admits several coverings of the size of $G$.

Whether this characterization admits an equivalent formulation in terms of more elementary graph properties is left open in this paper.

\subsection{Further related work}

The two problems considered here, namely \emph{spanning tree construction} and \emph{topology recognition}, are well known problems in the distributed computing literature.
A series of papers initiated by P. Rosenstiehl \cite{Rosen} and motivated by parallel computing,
studies the recognition of graph families by graph automata with a distinguished vertex~\cite{R75,WR79,R94,PR02}. In these articles, the goal is to recognize the \emph{family} of regular topology the automata lives in, such as rectangular structures, torus, cylinder, moebius band, and sphere.
Angluin studies in \cite{Angluin} the problem of recognizing various graph properties by a distributed algorithm in the same coarse-grain interaction model as for election and using the same type of coverings.
In a related model, Courcelle and M\'etivier prove using coverings that a certain minor-closed class of graphs cannot be recognized in general~\cite{CMminors}. The impact of prior knowledge in property recognition was studied, among others, in~\cite{GMMrecog}.
The topology recognition problem in the form we consider in this paper, was introduced by Yamashita and Kameda in~\cite{YKsolvable} (preliminary version~\cite{YK88}).
More recently, \cite{DP16}~considers a variant of the topology recognition problem that lies between identified and anonymous networks, in a message passing model. In this version, each node is given a color which is not necessarily unique, but a color must have a bounded multiplicity $k$, in terms of which the feasibility of election and topology recognition are characterized. 
The topology recognition problem has also been considered from other perspectives than symmetry, for example in terms of its round complexity in the CONGEST model (with unique identifiers)~\cite{CFSV16,FRST16}. More remotely, though still related, the (centralized) problem of assigning port numbers to a symmetric graph in a way that allows fast symmetry breaking in synchronous message passing models was studied in~\cite{KKP16}.

While the problem of constructing a \emph{spanning tree} has been extensively studied in the distributed computing literature, only a tiny fraction of this literature addresses the \emph{feasibility} of this problem. The most notable contribution, here again, is the one of Yamashita and Kameda in~\cite{YKsolvable}. 
Among other things, they prove that this problem is equivalent to the problem of \emph{electing an edge} in the asynchronous message passing model with port numbers, and they give a characterization of graphs where these problems are solvable based on the concepts of views
and symmetricity, to which we shall return later.

Finally, the general literature of coverings include a few other classical pieces. In particular, Leighton~\cite{L82} characterized pairs of graphs that share a common finite covering. Amit et al.~\cite{random-lifts} initiated a series of works where basic properties of random coverings are studied (such as $k$-connectivity). Fianlly, the reader is referred to Gross and Tucker~\cite{GT} for further general background on graph coverings.

\section{Network Model and Basic Definitions}

We consider the usual asynchronous message passing
model (\cite{YKsolvable,Tel}). A network is represented by a simple
connected anonymous graph $G=(V(G),E(G))$ (or simply $(V,E)$ when no ambiguity arises) where
vertices correspond to processes and edges to communication links. 
The vertices that share an edge with a vertex $v$ are called the \emph{neighbors} of $v$, noted $N_G(v)$. The edges incident to $v$ are noted $I_G(v)$ and the number of such edges is called the degree $\deg_G(v)$ of $v$.
Each process can distinguish between the edges
that are incident to it using \emph{port numbers}, i.e., for each $u \in V(G)$ there exists a
bijection $\delta_u$ between these edges and
$[1,\deg_G(u)]$. The set $\delta$ of functions
$\{\delta_u\mid u\in V(G)\}$ is called a \emph{port numbering} of $G$. We denote by $(G,\delta)$ the
 graph $G$ with port numbering~$\delta$.

Each process $v$ in the network represents an entity that is capable
of performing computation steps, sending messages via some port and
receiving any message via some port that was sent by the corresponding
neighbor. We consider asynchronous systems, i.e., each message
 may take an unpredictable (but finite) amount of
time to be delivered. Note that we consider only reliable systems: no fault can occur
on processes or communication links. We also assume that the channels
are FIFO, i.e., for each channel, the messages are delivered in the
order they have been sent.
In this model, a distributed algorithm is given by a local algorithm
that all the nodes execute (they all have
the same). This algorithm consists of a sequence of
computation steps interspersed with instructions to send and receive messages.

\subsection{Spanning tree construction and topology recognition}
\label{sec:problems}

We study the role of coverings in the context of two problems, which are the spanning tree construction problem and the topology recognition problem (also called ``find topology'' in~\cite{YKsolvable}). For both problems, we consider the above asynchronous message passing setting and we assume that the nodes know the total number $n=|V|$ of nodes in the network. The goal is to characterize the graphs that admit a solution
for all port numberings. More precisely, the definitions are as follows:

\begin{itemize}
  \item \textit{Spanning Tree Construction:} 
    Given a communication graph $G=(V,E)$, knowing $n=|V|$, 
    each node must select a set of incident edges, whose global union defines a spanning tree.\smallskip
  \item \textit{Topology Recognition:} 
    Given a communication graph $G=(V,E)$, knowing $n=|V|$,
    each node must compute a graph $G'$ isomorphic to $G$.
\end{itemize}


As observed in~\cite{YKsolvable},
the problem of constructing a spanning tree is equivalent, in terms of feasibility, to the problem of distinguishing a vertex or a pair of neighbor vertices (election of a leader or co-leader). We recall the main arguments of this equivalence, as it plays a central role in the characterization of feasible cases. The reader is referred to Appendix~\ref{sec:coleader} if more details are needed.
The equivalence can be established through two algorithms---one in each direction---which solve one problem provided that the other problem is already solved. The first algorithm shows that one can always elect a leader or a co-leader in a tree, through a sequence of recursive leaf-elimination. This principle was independently rediscovered many times, and it appears in the context of Theorem~4.4 of~\cite{Angluin} (in a different model). A formal description in the asynchronous message passing model is given by Tel in~\cite{Tel} p. 192-193. Initially, each vertex waits until it
has received a message from all but one of its neighbors, then it sends a message to the remaining neighbor. As a special case, the leaves send a message to their unique neighbor without waiting. Then, either a vertex eventually receives a message from all its neighbors before sending a message, and it becomes $leader$, or two neighbors eventually receive a message from each other and they become $co$-$leaders$.

The second algorithm builds a tree initiated at a leader or at two co-leaders. Such an algorithm was proposed by Tarry~\cite{Tarry} (a formal description is given in~\cite{Tel} p. 207). It can be 
formulated through two rules: 1) a process never forwards the token twice through the same channel; and 2) a non-$leader$ (or non-$co$-$leader$) forwards the token to the
neighbor from which it first received it only if there is no other channel possible according to the first rule. In spirit, this process performs a depth-first search from the leader, or two non-overlapping DFSs from the co-leaders, and the edges through which a token is received for the first time become part of the tree. Termination can be made explicit by adding a standard recursive component that propagates notifications from the leaves back to the root(s).
In conclusion:
\begin{theorem}\label{1-2}
In terms of feasibility, computing a spanning tree is equivalent to distinguishing a vertex or a pair of neighboring vertices.
\end{theorem}

The characterization of feasible cases for topology recognition does not rely on such equivalences. It is based directly on graph coverings.

\section{A Quick Tour of Graph Covering Theory}

Generally speaking, a graph $G$ is a \emph{covering} (also called a \emph{lift}) of a graph $H$ if there is a map from the vertices of $G$ to the vertices of $H$ that induces a local isomorphism. Informally, the local configurations of the vertices are the same in $G$ and $H$. The exact definition depends on the type of graph considered, in particular whether the graph is directed or undirected, and simple or non-simple (\ie with loops and/or multiple edges). This section offers a quick overview of the topic, starting with the simplest setting.

\subsection{Coverings of undirected graphs}
\label{sec:undirected}
When the graph is simple and undirected, the local isomorphism induced by the homomorphism can be defined in terms of \emph{neighborhood} (as opposed to edges). Following \cite{Angluin,Bod},
we say that $G$ is a {\em covering} of $H$ via
$\varphi$  if $\varphi$ is a surjective
homomorphism from $G$ onto $H$ such that, for every vertex $v$ of
$V(G)$, the restriction of $\varphi$ to $N_G(v)$ is a bijection onto
$N_H(\varphi(v))$. In other words, neighbors are sent to neighbors. The covering is proper if $G$ and $H$ are not
isomorphic, as in Figure~\ref{fig:exemple}.
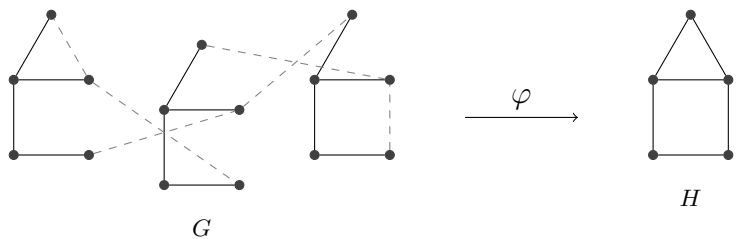
\begin{figure}[h]
    \begin{center}
    \begin{tikzpicture}[centered]
      \tikzstyle{every node}=[defnode]
      \path (0,0) node(a){};
      \path (a)+(1,0) node(b){};
      \path (a)+(0,1) node(c){};
      \path (a)+(1,1) node(d){};
      \path (c)+(60:1) node(e){};

      \path (a)+(2,-0.4) node(aa){};
      \path (aa)+(1,0) node(bb){};
      \path (aa)+(0,1) node(cc){};
      \path (aa)+(1,1) node(dd){};
      \path (cc)+(60:1) node(ee){};

      \path (aa)+(2,0.4) node(aaa){};
      \path (aaa)+(1,0) node(bbb){};
      \path (aaa)+(0,1) node(ccc){};
      \path (aaa)+(1,1) node(ddd){};
      \path (ccc)+(60:1) node(eee){};

      \path (8.5,0) node(bota){};
      \path (bota)+(1,0) node(botb){};
      \path (bota)+(0,1) node(botc){};
      \path (bota)+(1,1) node(botd){};
      \path (botc)+(60:1) node(bote){};

      \tikzstyle{every path}=[]
      \draw (a)--(b);
      \draw (a)--(c);
      \draw (c)--(d);
      \draw (c)--(e);

      \draw (aa)--(bb);
      \draw (aa)--(cc);
      \draw (cc)--(dd);
      \draw (cc)--(ee);

      \draw (aaa)--(bbb);
      \draw (aaa)--(ccc);
      \draw (ccc)--(ddd);
      \draw (ccc)--(eee);

      \draw (bota)--(botb);
      \draw (bota)--(botc);
      \draw (botc)--(botd);
      \draw (botc)--(bote);
      \draw (botb)--(botd);
      \draw (botd)--(bote);

      \tikzstyle{every node}=[]
      \draw (6,0.5) edge[->] node[above]{\large $\varphi$} (7.5,.5);

      \tikzstyle{every path}=[dashed, gray]
      \draw (b)--(dd);
      \draw (d)--(bb);

      \draw (ee)--(ddd);
      \draw (dd)--(eee);

      \draw (d)--(e);

      \draw (bbb)--(ddd);

      \tikzstyle{every node}=[black]
      \path (aa)+(0.5,0) node[below=10pt] (g){$G$};
      \path (bota)+(0.5,0) node[below=10pt] (h){$H$};
    \end{tikzpicture}
    \end{center}
    \caption{\label{1stex}\label{fig:exemple}
      Example of a 3-sheeted covering of a simple undirected graph, where the morphism $\varphi$ sends the vertices of $G$ to the vertices drawn similarly in $H$.
    }
\end{figure}

This definition can be extended to the case that
 multiple edges are allowed, by substituting $N_G(v)$ by $I_G(v)$,
in the definition. Thus, for every vertex $v$ of
$V(G)$ the restriction of $\varphi$ to $I_G(v)$ is a bijection onto
$I_H(\varphi(v))$, \ie incident edges are sent to incident edges (see Figure~\ref{premiersexemples}).
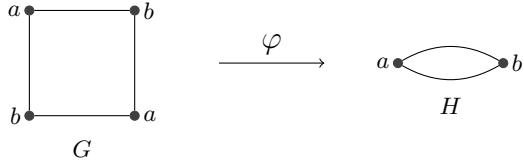
\begin{figure}[h]
\begin{center}
    \begin{tikzpicture}[scale=1.4]
      \tikzstyle{every node}=[defnode]
      \path (3.5,3.5) node (a){};
      \path (4.5,3.5) node (b){};
      \path (1.8,3.5) coordinate (bot);
      \path (2.8,3.5) coordinate (top);
      \path (0,3) node (c){};
      \path (1,3) node (d){};
      \path (0,4) node (e){};
      \path (1,4) node (f){};
      \tikzstyle{every node}=[]
      \path (a) node[left] (la){$a$};
      \path (b) node[right] (la){$b$};
      \path (c) node[left] (la){$b$};
      \path (d) node[right] (la){$a$};
      \path (e) node[left] (la){$a$};
      \path (f) node[right] (la){$b$};
      \draw (bot) edge[->] node[above]{\large $\varphi$} (top);
      \draw (a) edge[bend left] (b);
      \draw (a) edge[bend right] (b);
      \draw (c)--(d)--(f)--(e)--(c);
      \path (c)+(0.5,0.1) node[below=10pt] (g){$G$};
      \path (a)+(0.5,0) node[below=10pt] (h){$H$};
    \end{tikzpicture}
\end{center}
\caption{Example of a $2$-sheeted covering of an undirected graph with multiple edges, where the two vertices labeled $a$ (resp. $b$) on $G$ are sent to the same vertex labeled $a$ (resp. $b$) in $H$.}
\label{premiersexemples}
\end{figure}
Given a covering map $\varphi$ from $G$ to $H$, we call $G$ the \emph{total graph} of $\varphi$ and $H$ the \emph{base} of $\varphi$ (terminology from~\cite{BVfibrations}). We say that a graph $G$ is \emph{minimal} for the covering relation if there does not exist a graph $H$ such that $G$ is a proper covering of $H$. Naturally, the minimality of a graph is dependent on the type of restriction considered. For example, the graph $G$ in Figure~\ref{premiersexemples} is minimal in the context of simple graphs, but not if multiple edges are allowed. Similarly, the base graph of $\varphi$ in Figure~\ref{fig:exemple} is minimal in the context of simple graphs, and even when multiple edges are allowed, but it is not minimal if directed loops are allowed. 

A key property of coverings, which is independent from the restriction considered, is that if $G$ is a covering of $H$ via $\varphi$, then the number of vertices (and edges) of $G$ is a multiple of the one in $H$. In fact, the following property holds:

\begin{lemma}[\cite{Reidemeister,Angluin}]
  \label{lem:tree}
Let $T$ be
  a subgraph of $H.$ If $T$ is a tree then $\varphi^{-1}(T)$ is a set
  of disjoint trees, each isomorphic to $T.$
\end{lemma}

Let $G$ be a covering of $H$ through $\varphi$.
By considering a spanning tree of $H$ and paths between any two vertices of $H$,
Lemma~\ref{lem:tree} implies that there exists an integer $q$ such that $Card(\varphi^{-1}(v))=q$ for
  all $v\in V(H)$. This number is called the number of {\em sheets} of the covering and $G$ (or $\varphi$) is called a {\em $q$-sheeted covering} of $H$.
A deep result by Reidemeister~\cite{Reidemeister} (through~\cite{Bod} p.~168) establishes that \emph{all} the $q$-sheeted coverings of a given graph $H$ can be obtained by making $q$
copies of an \emph{arbitrary} spanning tree $T$ of $H$.
More precisely,


\begin{theorem}[\cite{Reidemeister}]
  Let $H$ be a connected graph and let $T$ be
  a spanning tree of $H$. A graph $G$ is
  a covering of $H$ if and only if there exist a non-negative integer
  $q$ and a set 
  $\Sigma=\{\sigma_{(u,v)}\mid\ u,v\in V(H), \{u,v\}\in E(H)\setminus
  E(T)\}$ of
  permutations\footnote { with the convention that $\sigma_{(u,v)} =
    \sigma^{-1}_{(v,u)}$ }  on $[1,q]$ such that $G$ is isomorphic to the graph
  $H_{T,\Sigma}$ defined by:
  \begin{eqnarray*}
    V(H_{T,\Sigma}) & = & \{ u_i\mid\ u\in V(H)\mid i\in [1,q] \},\\
    E(H_{T,\Sigma}) & = & \{\{u_i,v_i\}\mid \{u,v\}\in E(T),\ i\in [1,q] \} \cup\\
    & & \{\{u_i,v_{\sigma_{(u,v)}(i)}\}\mid  \{u,v\}\in
    E(H)\setminus E(T),\ i\in [1,q] \}. 
  \end{eqnarray*}
\end{theorem}

The imprint of Reidemeister's construction is visible in Figure~\ref{fig:exemple}, where plain edges depict the copies of a spanning tree of the base graph, and dashed edges correspond to a possible choice of permutations for the remaining edges (using $(1)(23)$ for the higher edge and $(12)(3)$ for the lower edge). Although the construction is presented here on simple undirected graphs, it is actually not specific; for instance, if $H$ has multiple edges, then at most one of them may belong to a spanning tree and permutations are chosen for all the others. We shall later apply Reidemeister's theorem to more general graphs.

\subsection{Coverings of directed graphs}

The symmetries arising in the asynchronous message passing model are captured by a special case of coverings in directed graphs (digraphs) possibly having multiple arcs and loops (see e.g.~\cite{BVelection,BVfibrations,Cthese,CMasynj}). The communication graph in this case is better seen as a directed symmetric graph $Dir(G)$ canonically obtained from $G$ by replacing each edge $\{u,v\}$ by two
arcs $(u,v)$ and $(v,u)$.
As observed by Lynch~\cite{Lynch},  
a distributed algorithm can be seen 
as executed indifferently on an undirected graph $G$ or on the directed graph $Dir(G)$.

The definition of coverings in general digraphs requires a certain level of expressivity. In the following, a digraph $D$ will thus be represented as a quadruplet $(V, A, s, t)$, where $V$ is the set of vertices, $A$ the set of arcs, and $s$ and $t$ are two maps that assign to an arc of $A$ a source and a target, respectively (both are vertices in $V$). A loop is an arc whose source and target coincide. Multi-arcs are arcs which have the same source and the same target. A digraph without loops or multi-arcs is called simple (for example, $Dir(G)$ is a simple digraph).
A \emph{symmetric} digraph $D$ is a digraph endowed with an involution (a self-inverse bijection) 
$Sym: A \rightarrow A$ such that for
every $a \in A, s(a)=t(Sym(a))$. In a symmetric digraph $D$, the
degree of a vertex $v$ is $\deg(v) = |\{a \mid s(a)=v\}|$ (or equivalently, $|\{a \mid
t(a) = v\}|$, since the graph is symmetric) and we denote by $N(v)$ the corresponding set of neighbors. The arcs whose target (resp. source) is $v$ are said to be \emph{incoming} at $v$ (resp. \emph{outgoing} from $v$).
These notations are parameterized (or subscripted) by the name of the graph in case of ambiguity.


A \emph{homomorphism} $\varphi$ between a digraph $D$ and a digraph
$D'$ is a mapping $\varphi \colon V(D)\cup A(D) \rightarrow V(D')\cup
A(D')$ such that $\varphi(s(a)) =
s(\varphi(a))$ and $\varphi(t(a)) = t(\varphi(a))$ for all arcs $a \in A(D)$.  In other words, incoming arcs at a vertex are sent to incoming arcs at the image of that vertex, and the same holds for outgoing arcs. A homomorphism that induces an \emph{isomorphism} (a bijection) between the incoming arcs of each vertex and the incoming arcs of its image is called a \emph{fibration}. The analog for outgoing arcs is an \emph{opfibration}. A \emph{covering} for digraphs is a homomorphism that is both a fibration and an opfibration. 
Finally, a \emph{symmetric covering} is a covering between
symmetric digraphs that preserves the involution $Sym$, that is, $\varphi(Sym(a)) = Sym(\varphi(a))$ for all arcs $a \in A(D)$. 
(See~\cite{BVfibrations} for additional background.)

Observe that not all fibrations are coverings, and not all coverings are \emph{symmetric} coverings, which may be a source of confusion. An example of covering that is not symmetric is shown in Figure~\ref{fig:non-symmetric}. This can be shown using basic properties of coverings. In particular, Theorem~38 in~\cite{BVfibrations} establishes that the coverings of $d$-bouquets are exactly the $d$-regular graphs, thus $G$ (which is $3$-regular) is a covering of $H$ (which is a $3$-bouquet). However, the fact that $3$ is an odd number forces at least one of the loops in $H$ to be its own symmetric. The inverse images of this loop in $G$ must be a perfect matching, because all the vertices in $G$ are inverse images of the only vertex in $H$. However, $G$ does not admit a perfect matching and $\varphi$ is therefore not a symmetric covering. (This type of argument was already used in~\cite{YKsolvable}.) 

\begin{figure}[h]
  \centering
\begin{tikzpicture}
\tikzstyle{every node}=[defnode]
\path (0,0) node (m){};
\path (-1,0) node (a){};
\path (0,1) node (b){};
\path (1,0) node (c){};
\draw (m)--(a);
\draw (m)--(b);
\draw (m)--(c);

\path (2,0) node (c1){};
\path (3,0) node (c2){};
\path (2.5,-1) node (c3){};
\path (2.5,1) node (c4){};
\draw (c1)--(c2);
\draw (c3)--(c);
\draw (c4)--(c);
\draw (c3)--(c1);
\draw (c4)--(c1);
\draw (c3)--(c2);
\draw (c4)--(c2);

\path (0,2) node (b1){};
\path (0,3) node (b2){};
\path (-1,2.5) node (b3){};
\path (1,2.5) node (b4){};
\draw (b1)--(b2);
\draw (b3)--(b);
\draw (b4)--(b);
\draw (b3)--(b1);
\draw (b4)--(b1);
\draw (b3)--(b2);
\draw (b4)--(b2);

\path (-2,0) node (a1){};
\path (-3,0) node (a2){};
\path (-2.5,-1) node (a3){};
\path (-2.5,1) node (a4){};
\draw (a1)--(a2);
\draw (a3)--(a);
\draw (a4)--(a);
\draw (a3)--(a1);
\draw (a4)--(a1);
\draw (a3)--(a2);
\draw (a4)--(a2);

\path (7,0) node (mm){};
\tikzstyle{every path}=[semithick,->]
\draw (mm) edge[loop,in=125,out=55,min distance=16mm,looseness=10] (mm);
\draw (mm) edge[loop,in=5,out=-65,min distance=16mm,looseness=10] (mm);
\draw (mm) edge[loop,in=-115,out=-185,min distance=16mm,looseness=10] (mm);

\tikzstyle{every node}=[]
\path (4.5,0) node(arr){\Large $\longrightarrow$} node[above=2pt]{$\varphi$};
\path (0,0) node[below=20pt] {$G=Dir(G)$};
\path (7,0) node[below=20pt] {$H$};
\end{tikzpicture}
\caption{\label{fig:non-symmetric}Example of a covering that is not symmetric. (For simplicity, $G$ is depicted instead of $Dir(G)$.)}
\end{figure}
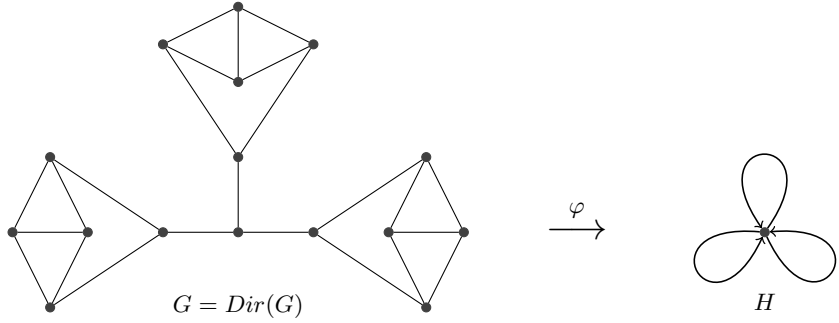

As already mentioned, symmetric coverings in symmetric directed graphs where the base graph may have multiple arcs and loops is the appropriate model to capture symmetries in the asynchronous message passing. More specifically, it is the extension of these concepts to the case where edges (arcs) are locally identified by port numbers (see below). To get an idea as to why digraphs are required, consider the basic example on the left side of Figure~\ref{arete}. Here, the communication graph $G$ consists of a single edge (say, with port numbers $1$ on both sides). Definitely, this graph is minimal if one restricts the base graph to be simple, and yet, election is not solvable in this graph. The adaptation of Angluin's lifting argument in such a context requires a directed loop. (The quotient graphs in~\cite{YKsolvable} use an undirected loop instead, with complications in the arguments of the proofs.)

\begin{figure}[h]
\begin{center}
  \begin{tikzpicture}[yscale=1.6]
    \tikzstyle{every node}=[defnode]
    \path (1,0) node (e){};
    \path (0,0) node (d){};
    \path (3,0) node (c){};
    \path (2,0) node (b){};
    \path (5,0) node (a){};
    \draw (d) edge (e);
    \tikzstyle{every path}=[semithick]
    \draw (b) edge[->,bend left=15] (c);
    \draw (c) edge[->,bend left=15] (b);
    \draw (a) edge[semithick,->,min distance=8mm,in=-30,out=30,looseness=10] (a);
    \tikzstyle{every node}=[]
    \path (.5,0) node[below=10pt](la){$G$};
    \path (2.5,0) node[below=10pt](la){$Dir(G)$};
    \path (5.25,0) node[below=10pt](la){$H$};
    \path (.5,.3) coordinate(to);
    \path (.5,.7) coordinate(from);
    \draw (1.5,0) node{$=$};
    \draw (3.5,0) edge[->] node[above]{$\varphi$} (4.5,0);
    \path (0,-.8) coordinate (bidon);
  \end{tikzpicture}
  \hspace{2cm}
    \begin{tikzpicture}
      \path (0,-.9) coordinate (bidon);
      \tikzstyle{every node}=[defnode]
      \path (0,0) node (a){};
      \path (1,0) node (b){};
      \path (0,1) node (c){};
      \path (1,1) node (d){};
      \draw (a)--(b)--(d)--(c)--(a);
      \tikzstyle{every node}=[]
      \path (.5,0) node[below=5pt] (lab){$G'$};
      \path (1.5,.5) node {$=$};
    \end{tikzpicture}
    \begin{tikzpicture}
      \path (0,-.88) coordinate (bidon);
      \tikzstyle{every node}=[defnode]
      \path (0,0) node (a){};
      \path (1,0) node (b){};
      \path (0,1) node (c){};
      \path (1,1) node (d){};
      \tikzstyle{every path}=[->,bend left=15,shorten >=.5pt]
      \draw (a) edge (b);
      \draw (b) edge (a);
      \draw (b) edge (d);
      \draw (d) edge (b);
      \draw (d) edge (c);
      \draw (c) edge (d);
      \draw (c) edge (a);
      \draw (a) edge (c);
      \tikzstyle{every node}=[]
      \path (.5,0) node[below=5pt] (lab){$Dir(G')$};
      \draw[->] (1.5,.5) -- node[above]{$\varphi'$} (2.5,.5);
      \tikzstyle{every node}=[defnode]
      \path (3.5,0.5) node (ef){};
      \tikzstyle{every path}=[]
      \draw (ef) edge[semithick,->,loop,min distance=10mm,in=-30,out=30,looseness=10] (ef);
      \draw (ef) edge[semithick,->,loop,min distance=10mm,in=150,out=210,looseness=10] (ef);
      \tikzstyle{every node}=[]
      \path (3.5,0.5) node[below=15pt] (lab){$H'$};
    \end{tikzpicture}
\end{center}
\caption{
  Two examples of (proper) symmetric coverings in symmetric digraphs.
}
\label{arete}
\end{figure}
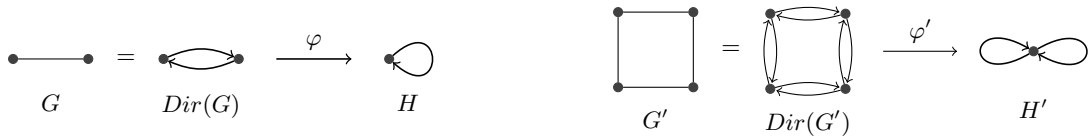






\subsection{Coverings and port numbers}

The above definitions of coverings can be naturally extended to the case that the digraph is equipped with a port numbering $\delta$. In this case, the additional requirement is that the local isomorphism induced by the homomorphism $\varphi$ should preserve local port numbers as well.
Two importants facts with port numberings are that (1) port numbers cannot create symmetries which do not pre-exist in the underlying graph; and (2) there always exist port numbers that preserve the symmetries of the underlying graph. The first is particularly useful to prove sufficient conditions, and the second to prove necessary conditions. In particular, in the context of symmetric digraphs, we will use the following:

\begin{lemma}[\cite{CMasynj}*]
  \label{lem:ports1}
  If $Dir(G)$ is minimal for the symmetric covering relation, then so is $(Dir(G),\delta)$ for all port numberings $\delta$.  
\end{lemma}

\begin{lemma}[\cite{CMasynj}*]
  \label{lem:ports2}
If $Dir(G)$ is a proper symmetric covering of a
symmetric digraph $D$, then there
exist port numberings $\delta$ for $G$ and $\delta'$ for $D$ such that $(Dir(G),\delta)$ is a
proper symmetric covering of $(D,\delta')$. 
\end{lemma}

\noindent
\textbf{*} These facts are established (among others) within the proof of Theorem~4.2 in~\cite{CMasynj}.\smallskip

\subsection{Mazurkiewicz's Algorithm}
\label{sec:mazurkiewicz}
In this section, we present an algorithm that plays a central role in the interplay between covering theory and distributed computing; yet, it seems that this algorithm is not so well known in the distributed computing community. Among other things, Mazurkiewicz's algorithm makes it possible for the nodes of a network to compute a graph $H$ of which the actual communication graph $G$ is a covering. This algorithm exists in various versions~\cite{Cthese} and was adapted by Chalopin and M\'etivier~\cite{CMasynj} to the case of \emph{symmetric} coverings in the asynchronous message passing model with port numbering. Although it is close in spirit to the computation of views (see e.g.~\cite{compressed}), Mazurkiewicz's algorithm has unique features which make it a powerful building block in other algorithms. In particular, the knowledge of (a bound on) $n$ is only required for explicit termination (whereas the computation of views crucially depends on this parameter to control the size of messages), and perhaps more importantly, the algorithm takes advantage of the asynchrony in the execution, yielding \emph{possibly} more information than the views when symmetries are broken by the communications.

The original version of this algorithm was presented as an enumeration algorithm by Mazurkiewicz~\cite{MazurEnum}, with the purpose of assigning to every node a unique identifier (integer) in a contiguous interval, starting at $1$. This algorithm was formulated in a coarse-grain computational model where a step of computation rewrites atomically the state of a node and that of its entire neighborhood. Mazurkiewicz establishes that the algorithm is guaranteed to succeed if and only if the communication graph $G$ is minimal for the covering relation in the context of simple undirected graphs (see Section~\ref{sec:undirected}). The adaptation of Mazurkiewicz's algorithm in the asynchronous message, hereafter denoted as $\mk$, follows the same lines as the original algorithm. We review here the main properties of $\mk$, in particular the properties useful to the problems considered in this paper. The reader is referred to Appendix~\ref{sec:mazurkiewicz-full} for more details or to the original articles~\cite{CMasynj,CGM12}.

Thus, $\mk$
is a distributed enumeration algorithm in the asynchronous message passing model (with port numbering) whose success is guaranteed in graphs which are minimal for the symmetric covering relation. 
The algorithm is mainly based on exchanging and comparing local configurations among neighbors and collecting everyone's local configurations. The definition of an order between configurations is used by the nodes to pick tentative identifiers which are later changed if a competitor with smaller configuration is discovered. 
The algorithm always terminates implicitly (\ie messages are no longer exchanged, but the nodes are not aware of that), and so, without any knowledge. The termination can be made explicit if the nodes know $n$ (or an upper bound on $n$)~\cite{CGM12}. Upon termination, the collection of configurations define a directed symmetric graph $D$ such that (1) $Dir(G)$ is a symmetric covering of $D$, and (2) $D$ may be any symmetric directed graph of which $Dir(G)$ is a symmetric covering, depending on the particular ordering of events in the execution. As a particular case, if one is lucky, then $D$ may happen to be equal to $Dir(G)$ itself (even if $Dir(G)$ is not minimal). At the other extreme, if one is unlucky, $D$ may be as small as Yamashita and Kameda's quotient graph (obtained from the views). More generally, $D$ may be any intermediate graph in the hierarchy of coverings below $Dir(G)$.



On the complexity side, an execution of
$\mk$ on $G$ takes 
$O(n^2 Diam(G))$ asynchronous rounds and $O(m^2n)$ messages of $O(\Delta \log n)$ bits (Prop.~4.4 in~\cite{CMasynj}), where $n$ is the number of vertices, $m$ the number of edges, $\Delta$ the maximum degree, and $Diam(G)$ the diameter of the network. The memory needed by each node is $O(\Delta n
\log n)$ bits. In some sense, the complexity of $\mk$ is ``natively'' polynomial, whereas that of the computation of views requires truncation and non-trivial compression techniques to become so (see~\cite{Norris} and~\cite{compressed}, respectively).

\section{Feasible Cases for Spanning Tree Construction in terms of Coverings}

In this section, we formulate necessary and sufficient conditions for the feasibility of constructing a spanning tree, provided that the nodes know (an upper bound on) the total number of nodes~$n$. The conditions are expressed in terms of graph coverings.
The added value of using Mazurkiewicz's algorithm together with the proper type of coverings is that these proofs are significantly simpler than their analogs using views and quotient graphs (even though the arguments rely on the same essential properties).

\begin{theorem}
  \label{un}
  \label{lem:spanning-sufficient}
  Let $G$ be the communication graph.
  There exists a distributed algorithm which computes a spanning tree of $G$ for all port numberings $\delta$ if and only if either $Dir(G)$ is minimal for the symmetric covering relation, or $Dir(G)$ is a $2$-sheeted symmetric covering of a symmetric graph $D$ having at least one loop and it is not a $q$-sheeted symmetric covering of a symmetric digraph with $q>2$.
\end{theorem}
\begin{proof}
  The proof relies on the equivalence between spanning tree construction and the problem of electing a leader or two coleaders (see Theorem~\ref{1-2} in Section~\ref{sec:problems}), combined with algorithm $\mk$ (see Section~\ref{sec:mazurkiewicz}).
  
  \emph{Sufficient side:}
  If $Dir(G)$ is minimal for the symmetric covering relation, then so is $(Dir(G),\delta)$ for all $\delta$ (Lemma~\ref{lem:ports1}), thus an execution of $\mk$ results in a successful enumeration of the vertices. In particular, the (unique) vertex with identifier $1$ can be considered as $leader$ in $G$.
If $Dir(G)$ is a $2$-sheeted symmetric covering of a symmetric digraph $D$ that has at least one loop
and it is not a $q$-sheeted symmetric covering of a symmetric digraph with $q>2$, then either $(Dir(G),\delta)$ is minimal (i.e. $\delta$ breaks the symmetry) and the previous argument applies, or $(Dir(G),\delta)$ is a $2$-sheeted symmetric covering of a digraph $(D,\delta')$ that has at least one loop (Lemma~\ref{lem:ports2}). In the latter case, $D$ must be minimal because $Dir(G)$ is not a $q$-sheeted symmetric covering of a symmetric digraph with $q>2$. Depending on communication asynchrony, either $\mk$ will enumerate the vertices of $G$ (and a $leader$ is again successfully distinguished), or it enumerates the vertices of $D$ with multiplicity $2$. As $D$ contains at least one loop (say, on vertex $v$), there must exist two neighbor vertices in $G$ (the inverse images of $\varphi^{-1}(v)$ of $v$) which end up with the same identifier. Recalling that $\mk$ outputs the same identified graph at all the nodes, it suffices to choose
 as $co$-$leaders$ the two neighbors in $G$ that $\varphi$ sends to the smallest such vertex in $D$.


\emph{Necessary side:}
Using the equivalence in the other direction, the goal is to prove that if the condition is not satisfied, then it is impossible to elect a single leader or a single pair of coleader in $G$.
If the condition is not satisfied, then $(Dir(G),\delta)$ is either a $q$-sheeted covering of a digraph $(D,\delta')$ with $q > 2$ or a $2$-sheeted symmetric covering of a \emph{loopless} digraph $(D,\delta')$ (both options are not mutually exclusive, but this fact is unimportant). In the first case, a lifting argument shows that the execution of the algorithm (whatever it be) may generate at least three $leader$ states (or three pairs of $co$-$leader$ states, in case of loops in $D$); in the second case, it may generate twice the $leader$ state (in a non-detectable way).
\qed
\end{proof}



\section{Feasible Cases for Topology Recognition in terms of Coverings}

In this section, we revisit the characterization of feasible cases for the topology recognition problem in terms coverings. Precisely, we characterize the conditions that a graph $G$ must satisfy to be recognizable unambiguously by a distributed algorithm. Analogs of these characterizations are given in~\cite{YKsolvable} in terms of views and quotient graphs, together with a sufficient condition in terms of coverings. The condition we present here is both necessary and sufficient in terms of coverings. The necessary side is established through lifting arguments, and the sufficient side follows from the powerful features of Mazurkiewicz's algorithm (presented in Section~\ref{sec:mazurkiewicz}). Some consequences of using this algorithm rather than the views are discussed. Finally, we construct a $28$-vertex graph that shows that Yamashita and Kameda's sufficient condition in terms of covering was indeed not necessary (as they expected themselves).

\subsection{Some knowledge is required}
\label{sec:knowledge}
Let us start by recovering, using a standard lifting argument, the necessity of considering additional knowledge at the nodes in order to make the problem solvable.
The proof adapts that of Theorem 5.5 in~\cite{Angluin} for simple undirected graphs in the context of Angluin's coarse-grained atomic model.

\begin{proposition}\label{noinformation}
  Let $D_1$ and $D_2$ be two symmetric
  digraphs. If $D_1$ is a proper symmetric 
covering of $D_2$, then no distributed algorithm can solve the topology recognition problem in both $D_1$ and $D_2$ without additional knowledge. 
\end{proposition}
\begin{proof}
Let $D_1$ be a proper symmetric covering of $D_2$ via $\varphi$.
Suppose that there exists a distributed algorithm $\cal A$ that recognizes 
successfully the topologies of $D_1$ and $D_2$.
By a lifting argument, there exists an execution where every step applied by $\cal A$ 
on a vertex $u$ of $D_1$ is also applied on $v=\varphi(u)$ and on the vertices of 
$\varphi^{-1}(v)$. As a result, there exists an execution of $\cal A$ that will recognize the same graph in $D_1$ and in $D_2$.
\qed
\end{proof}

\begin{corollary}
If $G$ is a graph such that
$Dir(G)$ is not minimal for the symmetric covering relation, then there
is no distributed
algorithm without additional knowledge for computing
the topology of $G$.
\end{corollary}

For these reasons, both Angluin~\cite{Angluin} and Yamashita and Kameda~\cite{YKsolvable} further assume that the nodes have limited knowledge of $G$, namely its number of vertices $n$. We follow them in this choice.

\subsection{A necessary and sufficient condition in terms of coverings}

Let $G$ be the actual communication graph and $Dir(G)$ be the corresponding symmetric digraph. We assume that the nodes know the number $n$ of participants in the network (vertices of $G$). The goal of the recognition problem is that the nodes compute a graph which is isomorphic to $Dir(G)$. We characterize below a necessary and sufficient condition to do this unambiguously, in terms of symmetric directed coverings. Intuitively, the nodes are susceptible of computing any graph $D$ of which $Dir(G)$ is a covering. (The exact argument considers port numbers as well.) The knowledge of $n$ allows them to avoid ambiguities related to the size of $Dir(G)$. But this is not always sufficient, as several non-isomorphic graphs of the same size may be symmetric directed covering of a same digraph $D$. The condition thus requires unicity of such a graph.

%

Starting with the necessary side, the following proposition states that if $Dir(G)$ is a 
$q$-sheeted symmetric
covering of a symmetric digraph $D$ such that $D$ has two non-isomorphic
$q$-sheeted symmetric coverings, then $G$ cannot be recognized.

\begin{proposition}
  \label{lem:topology-necessary}
Let $D_1$ and $D_2$ be two non-isomorphic connected symmetric digraphs. Let $q$ be a positive integer.
If there exists
a symmetric digraph $D$ 
 such that
$D_1$ and $D_2$ are $q$-sheeted symmetric coverings of $D$,
then there exists no distributed algorithm for solving the topology
recognition problem for all port numberings in $D_1$ and $D_2$, even if the number of vertices $n$ of $D_1$ and $D_2$ is known.
\end{proposition}
\begin{proof}
Let $D_1$ and $D_2$ be two non-isomorphic connected symmetric digraphs on $n$ vertices. 
Let $D$ be a symmetric digraph such that
$D_1$ and $D_2$ are $q$-sheeted symmetric coverings of $D$ via
$\varphi_1$ and $\varphi_2$, respectively. By Lemma~\ref{lem:ports2}, there exists three port numbering functions $\delta, \delta_1$, and $\delta_2$ that preserve these covering relations from $(D_1,\delta_1)$ and $(D_2,\delta_2)$ to $(D,\delta)$.
Suppose that there exists a distributed algorithm $\cal A$ that recognizes 
$(D_1,\delta_1)$. By a lifting argument,
there exists an execution such that, every time a step is applied by $\cal A$ 
on a vertex $u$ of $(D_1,\delta_1)$, the same step
is applied on $v=\varphi_1(u)$, on the vertices of 
$\varphi_1^{-1}(v)$ and on the vertices of $\varphi_2^{-1}(v)$.
As a result, there exists an execution of $\cal A$ in $(D_2,\delta_2)$ that wrongly recognizes $(D_1,\delta_1)$. 
\qed
\end{proof}

The sufficient side relies on Mazurkiewicz's algorithm,

\begin{proposition}
  \label{lem:topology-sufficient}
Let $G=(V,E)$ be a communication graph. 
If all the symmetric digraphs $D$ such that $Dir(G)$ 
is a $q$-sheeted symmetric covering
of $D$ admit a unique $q$-sheeted symmetric covering that is simple, connected, and has $n$ vertices, 
then there exists a distributed algorithm
solving topology recognition for all port numberings $\delta$ in $(G,\delta)$ with the knowledge of $n=|V|$.
\end{proposition}


\begin{proof}
  The proof follows from the properties of algorithm $\mk$. In particular, given $(G,\delta)$, $\mk$ outputs a symmetric digraph $(D,\delta')$ such that $(G, \delta)$ is a $q$-sheeted symmetric covering of $(D,\delta')$ for some $q$. Clearly, if $G$ is uniquely determined from $D$, then $G$ is also uniquely determined from $(D, \delta')$. What remains to be done is to show that the nodes can now effectively reconstruct $G$ such that $Dir(G)$ is a $q$-sheeted symmetric directed covering of $D$. Using reidemeister's construction (see Section~\ref{sec:undirected}), the nodes can locally enumerate all the $q$-sheeted symmetric coverings of $D$ and select the first one that is both simple and connected. Since this graph is guaranteed to be unique, it must be isomorphic to $G$. 
\end{proof}

Putting Propositions~\ref{lem:topology-necessary} and~\ref{lem:topology-sufficient} together, we obtain:

\begin{theorem}\label{deux}
Let $G=(V,E)$ be a communication graph.
There exists a distributed algorithm
solving the topology recognition problem for all port numberings $\delta$ in $(G,\delta)$ with the knowledge of $n=|V|$
if and only if any symmetric digraph $D$ such that $Dir(G)$ 
is a $q$-sheeted symmetric covering
of $D$ admits exactly one simple connected $q$-sheeted symmetric covering.
\end{theorem}

\begin{figure}[h]
  \centering
  \begin{tikzpicture}[xscale=2.5,yscale=3.5]
    \path (0.6,2) node (n8cross){
      \begin{tikzpicture}[scale=.9]
        \tikzstyle{every node}=[defnode]
        \path (0, 0) node (a){};
        \path (0, 1) node (b){};
        \path (1, 0) node (c){};
        \path (1, 1) node (d){};
        \path (2, 0) node (f){};
        \path (2, 1) node (e){};
        \path (3, 0) node (h){};
        \path (3, 1) node (g){};
        \draw (a) -- (c);
        \draw (e) -- (g);
        \draw (a) -- (g);
        \draw (b) -- (d) -- (f) -- (h) ;
        \draw (a) -- (b);
        \draw (c) -- (f);
        \draw (e) -- (d);
        \draw (g) -- (h);
      \end{tikzpicture}
    };
    \path (4,2) node (n8triangle){
      \begin{tikzpicture}[scale=.9]
        \tikzstyle{every node}=[defnode]
        \path (0, 0) node (a){};
        \path (0, 1) node (b){};
        \path (1, 0) node (c){};
        \path (1, 1) node (d){};
        \path (2, 0) node (e){};
        \path (2, 1) node (f){};
        \path (3, 0) node (g){};
        \path (3, 1) node (h){};
        \draw (a) -- (c) -- (e) -- (g) ;
        \draw (b) -- (d) -- (f) -- (h) ;
        \draw (a) -- (b);
        \draw (a) -- (d);
        \draw (e) -- (h);
        \draw (g) -- (h);
      \end{tikzpicture}
    };
    \path (2.4,2) node (n8bend){
      \begin{tikzpicture}[scale=.9]
        \tikzstyle{every node}=[defnode]
        \path (0, 0) node (a){};
        \path (0, 1) node (b){};
        \path (1, 0) node (c){};
        \path (1, 1) node (d){};
        \path (2, 0) node (e){};
        \path (2, 1) node (f){};
        \path (3, 0) node (g){};
        \path (3, 1) node (h){};
        \draw (a) -- (c);
        \draw (e) -- (g);
        \draw (b) -- (d);
        \draw (f) -- (h);
        \draw (a) -- (b);
        \draw (c) -- (d);
        \draw (e) -- (f);
        \draw (g) -- (h);

        \draw (a) edge[bend right] (e);
        \draw (d) edge[bend left] (h);
      \end{tikzpicture}
    };
    
    \path (1.1,1) node (n4loop){
      \begin{tikzpicture}
        \tikzstyle{every node}=[defnode]
        \path (0, 0) node (a){};
        \path (1, 0) node (b){};
        \path (0, 1) node (c){};
        \path (1, 1) node (d){};
        \draw (a) edge[thick,loopsw] (a);
        \draw (d) edge[thick,loopne] (d);
        \draw (a) -- (b);
        \draw (b) -- (d);
        \draw (d) -- (c);
        \draw (c) -- (a);
      \end{tikzpicture}
    };
    \path (3.7,1) node (n4simple){
      \begin{tikzpicture}
        \tikzstyle{every node}=[defnode]
        \path (0, 0) node (a){};
        \path (1, 0) node (b){};
        \path (0, 1) node (c){};
        \path (1, 1) node (d){};
        \draw (a) -- (b);
        \draw (a) -- (c);
        \draw (a) -- (d);
        \draw (b) -- (d);
        \draw (a) -- (d);
        \draw (c) -- (d);
      \end{tikzpicture}
    };
    \path (2.4,1) node (n4multi){
      \begin{tikzpicture}
        \tikzstyle{every node}=[defnode]
        \path (0, 0) node (a){};
        \path (1, 0) node (b){};
        \path (0, 1) node (c){};
        \path (1, 1) node (d){};
        \draw (a) -- (c);
        \draw (a) edge[bend left] (b);
        \draw (a) edge[bend right] (b);
        \draw (c) edge[bend left] (d);
        \draw (c) edge[bend right] (d);
      \end{tikzpicture}
    };
  
    \path (2.4,0) node (n2){
      \begin{tikzpicture}
        \tikzstyle{every node}=[defnode]
        \path (0, 0) node (a){};
        \path (1, 0) node (b){};
        \draw (a) edge[thick,loopleft] (a);
        \draw (a) edge[bend left] (b);
        \draw (a) edge[bend right] (b);
      \end{tikzpicture}
    };

    \tikzstyle{every path}=[]
    \path (-.1,1.5) coordinate (l);
    \path (4.9,1.5) coordinate (r);
    

    \tikzstyle{every node}=[below=-1pt,font=\normalsize]
    \path (n8cross.south) node {$G_5$};
    \path (n8triangle.south) node {$G_7$};
    \path (n8bend.south) node {$G_6$};
    \path (n4loop.south) node {$G_2$};
    \path (n4simple.south) node {$G_4$};
    \path (n4multi.south) node {$G_3$};
    \path (n2.south) node {$G_1$};
    
    \tikzstyle{every path}=[->, gray, thick, shorten >= 14pt, shorten <= 14pt]
    \draw (n8cross) edge (n4loop);
    \draw (n8triangle) edge (n4simple);
    \draw (n8bend) edge (n4simple);
    \draw (n8bend) edge (n4loop);
    \draw (n8bend) edge (n4multi);

    \draw (n4loop) edge (n2);
    \draw (n4simple) edge (n2);
    \draw (n4multi) edge[shorten <= 20pt] (n2);
  \end{tikzpicture}
  
  \caption{\label{fig:hierarchy} A partial hierarchy of symmetric coverings. For simplicity, pairs of symmetric arcs between two vertices are depicted as undirected edges (e.g. $G_3$ actually has $10$ arcs, not $5$). Non-connected coverings of $G_1$ are omitted and non-simple graphs in the third level ($n=8$) of the hierarchy are omitted as well. Finally, port numbers are not represented but they are chosen to preserve symmetries of the underlying graph (Lemma~\ref{lem:ports2}).}
\end{figure}
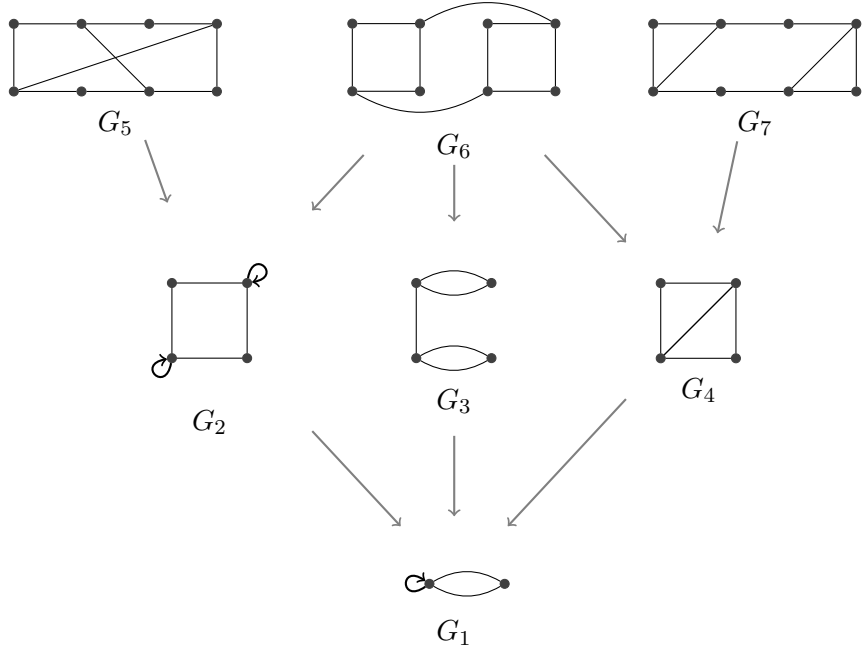

\subsubsection{Example scenarios.}
To illustrate the above characterizations and the special features of Mazurkiewicz's algorithm, let us consider one of the possible communication graphs $G$ and some of the graphs it covers in the hierarchy of Figure~\ref{fig:hierarchy}. Consider, for instance, that $G=G_4$. In this case, $\mk$ computes either $G_1$ or $G_4$ (depending on the effect of asynchrony). If $\mk$ computes $G_4$, then recognition is direct, given the fact that $n=4$ is known. If $\mk$ computes $G_1$, then $G$ is also successfully recognized because $G_4$ is the \emph{only} symmetric covering of $G_1$ on $4$ vertices which is simple and connected. In all the other cases, \ie if $G$ is $G_5, G_6,$ or $G_7$, then recognition may fail because all these graphs are symmetric coverings of $G_1$ (with the same number of sheets). Now, recall that algorithm $\mk$ is susceptible of computing \emph{any} of the graph $D$ such that $G$ is a symmetric covering of $D$, depending on the effect of asynchrony. If the turn of events is fortunate, recognition may still occur. In particular, $\mk$ may output the communication graphs themselves, or some intermediate graph that carries more information than $G_1$. For instance, if $\mk$ outputs $G_3$, whose unique (connected, simple) $2$-sheeted symmetric covering is $G_6$, then this graph will be recognized. This ability to profit from asynchrony is a crucial asset of covering-based characterizations, through Mazurkiewicz's algorithm. It differs from existing approaches based on the views, whose computation is made through semi-synchronous rounds which destroys the effect of asynchrony (resulting essentially to computing the analog of $G_1$ as a quotient graph in all executions).

\subsection{Non-necessity of Yamashita and Kameda's condition}

In~\cite{YKsolvable}, Yamashita and Kameda prove a sufficient condition for the topology recognition problem in terms of coverings. For completeness, the knowledge of $n$ is also assumed and the type of coverings is the same as in Angluin's work~\cite{Angluin}, namely coverings of simple undirected graphs.
The condition is as follows:

  \begin{theorem}[Theorem~23 of~\cite{YKsolvable}]\label{yk}
Given a graph $G$, if there is no graph $H$ with the same size as $G$ such that 
$G$ and $H$ have a finite common covering then there exists a distributed
algorithm for computing the topology of $G$.
\end{theorem}



Although Yamashita and Kameda were aware that this condition must not be necessary, no explicit graph has been constructed so far that illustrates this fact.
We present in Figure~\ref{G1G2} such a graph based on $28$ vertices (which may be the smallest possible counterexample).
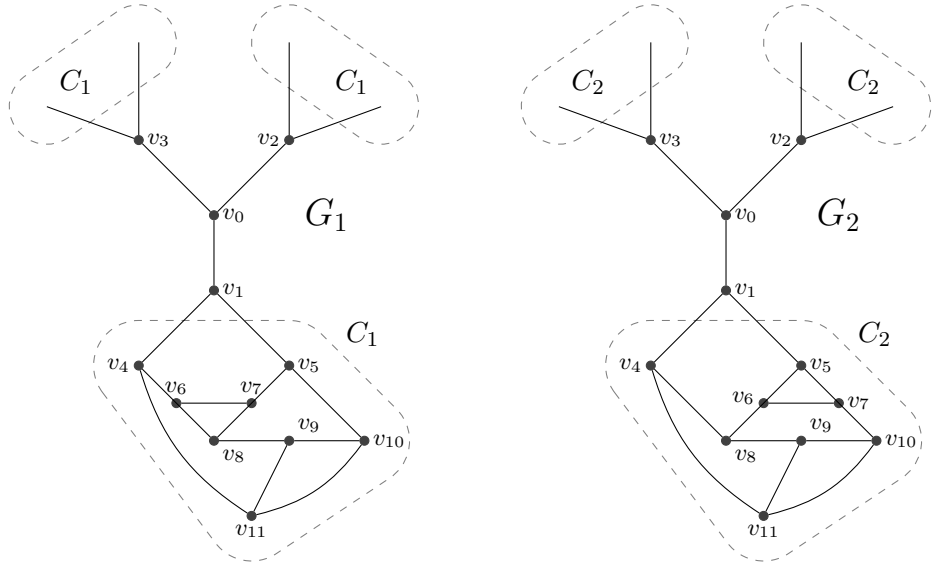
\begin{figure}[h]
\usetikzlibrary{fit,calc,arrows.meta}
\newcommand{\convexpath}[2]{
  [   
  create hullcoords/.code={
    \global\edef\namelist{#1}
    \foreach [count=\counter] \nodename in \namelist {
      \global\edef\numberofnodes{\counter}
      \coordinate (hullcoord\counter) at (\nodename);
    }
    \coordinate (hullcoord0) at (hullcoord\numberofnodes);
    \pgfmathtruncatemacro\lastnumber{\numberofnodes+1}
    \coordinate (hullcoord\lastnumber) at (hullcoord1);
  },
  create hullcoords
  ]
  ($(hullcoord1)!#2!-90:(hullcoord0)$)
  \foreach [
  evaluate=\currentnode as \previousnode using \currentnode-1,
  evaluate=\currentnode as \nextnode using \currentnode+1
  ] \currentnode in {1,...,\numberofnodes} {
    let \p1 = ($(hullcoord\currentnode) - (hullcoord\previousnode)$),
    \n1 = {atan2(\y1,\x1) + 90},
    \p2 = ($(hullcoord\nextnode) - (hullcoord\currentnode)$),
    \n2 = {atan2(\y2,\x2) + 90},
    \n{delta} = {Mod(\n2-\n1,360) - 360}
    in 
    {arc [start angle=\n1, delta angle=\n{delta}, radius=#2]}
    -- ($(hullcoord\nextnode)!#2!-90:(hullcoord\currentnode)$) 
  }
}
\begin{center}
\begin{tikzpicture}
  \tikzstyle{every node}=[draw,darkgray,circle,fill,inner sep=1.2pt]
  \path (-1,2) node (v3){};
  \path (1,2) node (v2){};
  \path (v3)+(160:1.3) coordinate (t1);
  \path (v3)+(90:1.3) coordinate (t2);
  \path (v2)+(90:1.3) coordinate (t3);
  \path (v2)+(20:1.3) coordinate (t4);
  \path (0,1) node (v0){};
  \path (0,0) node (v1){};
  \path (-1,-1) node (v4){};
  \path (1,-1) node (v5){};
  \path (-.5,-1.5) node (v6){};
  \path (.5,-1.5) node (v7){};
  \path (0,-2) node (v8){};
  \path (1,-2) node (v9){};
  \path (2,-2) node (v10){};
  \path (0.5,-3) node (v11){};
  \tikzstyle{every node}=[]
  \path (v0) node[right] {$v_0$};
  \path (v0)+(1.5,0) node {\Large $G_1$};
  \path (v1) node[right] {$v_1$};
  \path (v2) node[left] {$v_2$};
  \path (v2) node[above right=.7cm] {\normalsize $C_1$};
  \path (v3) node[right] {$v_3$};
    \path (v3) node[above left=.7cm] {\normalsize $C_1$};
  \path (v4) node[left] {$v_4$};
  \path (v5) node[right] {$v_5$};
  \path (v6) node[above] {$v_6$};
  \path (v7) node[above] {$v_7$};
  \path (v8) node[below right] {$v_8$};
  \path (v9) node[above right] {$v_9$};
  \path (v5) node[xshift=.5cm,above right=.2cm] {\normalsize $C_1$};
 \path (v10) node[right] {$v_{10}$};
  \path (v11) node[below] {$v_{11}$};

  \draw[gray,dashed] \convexpath{t1,t2}{.5cm};
  \draw[gray,dashed] \convexpath{t3,t4}{.5cm};
  \draw[gray,dashed] \convexpath{v4,v5,v10,v11}{.6cm};

  \draw (v3)--(t1);
  \draw (v3)--(t2);
  \draw (v2)--(t3);
  \draw (v2)--(t4);
  \draw (v3)--(v0);
  \draw (v2)--(v0);
  \draw (v0)--(v1);
  \draw (v1)--(v4);
  \draw (v1)--(v5);
  \draw (v4)--(v8);
  \draw (v5)--(v8);
  \draw (v6)--(v7);
  \draw (v8)--(v9);
  \draw (v9)--(v10);
  \draw (v5)--(v10);
  \draw (v9)--(v11);
  \draw (v11) edge[bend left=20] (v4);
  \draw (v11) edge[bend right=20] (v10);
\end{tikzpicture}
\hspace{1cm}
\begin{tikzpicture}
  \tikzstyle{every node}=[draw,darkgray,circle,fill,inner sep=1.2pt]
  \path (-1,2) node (v3){};
  \path (1,2) node (v2){};
  \path (v3)+(160:1.3) coordinate (t1);
  \path (v3)+(90:1.3) coordinate (t2);
  \path (v2)+(90:1.3) coordinate (t3);
  \path (v2)+(20:1.3) coordinate (t4);
  \path (0,1) node (v0){};
  \path (0,0) node (v1){};
  \path (-1,-1) node (v4){};
  \path (1,-1) node (v5){};
  \path (.5,-1.5) node (v6){};
  \path (1.5,-1.5) node (v7){};
  \path (0,-2) node (v8){};
  \path (1,-2) node (v9){};
  \path (2,-2) node (v10){};
  \path (0.5,-3) node (v11){};
  \tikzstyle{every node}=[]
  \path (v0) node[right] {$v_0$};
  \path (v0)+(1.5,0) node {\Large $G_2$};
  \path (v1) node[right] {$v_1$};
  \path (v2) node[left] {$v_2$};
    \path (v2) node[above right=.7cm] {\normalsize $C_2$};
  \path (v3) node[right] {$v_3$};
    \path (v3) node[above left=.7cm] {\normalsize $C_2$};
  \path (v4) node[left] {$v_4$};
  \path (v5) node[right] {$v_5$};
  \path (v6) node[left,yshift=2pt] {$v_6$};
  \path (v7) node[right] {$v_7$};
  \path (v8) node[below right] {$v_8$};
  \path (v9) node[above right] {$v_9$};
  \path (v5) node[xshift=.5cm,above right=.2cm] {\normalsize $C_2$};
  \path (v10) node[right] {$v_{10}$};
  \path (v11) node[below] {$v_{11}$};

 \draw[gray,dashed] \convexpath{t1,t2}{.5cm};
 \draw[gray,dashed] \convexpath{t3,t4}{.5cm};
  \draw[gray,dashed] \convexpath{v4,v5,v10,v11}{.6cm};

  \draw (v3)--(t1);
  \draw (v3)--(t2);
  \draw (v2)--(t3);
  \draw (v2)--(t4);
  \draw (v3)--(v0);
  \draw (v2)--(v0);
  \draw (v0)--(v1);
  \draw (v1)--(v4);
  \draw (v1)--(v5);
  \draw (v4)--(v8);
  \draw (v5)--(v8);
  \draw (v6)--(v7);
  \draw (v8)--(v9);
  \draw (v9)--(v10);
  \draw (v5)--(v10);
  \draw (v9)--(v11);
  \draw (v11) edge[bend left=20] (v4);
  \draw (v11) edge[bend right=20] (v10);
\end{tikzpicture}
\caption{The graphs $G_1$ and $G_2$ illustrate the fact that the sufficient
  condition given in Theorem 23 in \cite{YKsolvable} is not necessary.
  Dashed line in $G_1$ means the duplication of the graph $C_1$
  described below, same for $C_2$ in $G_2$.}
\label{G1G2}
\end{center}
\end{figure}
The two graphs $G_1$ and $G_2$ have the same size, and are non isomorphic.
Furthermore, they are $3$-regular thus they have a finite common
covering by Theorem 1  in \cite{Gardiner}.
Nevertheless,
we will show below that the graphs
$Dir(G_1)$ and $Dir(G_2)$ are minimal for the \emph{symmetric}
covering relation, which implies that a distributed algorithm like Mazurkiewicz's algorithm can recognize (and thus distinguish) them.

Using, among others, Reidemeister's construction (see Section~\ref{sec:undirected}), we will show that $Dir(G_1)$ is minimal for the symmetric covering relation (the same arguments hold for $Dir(G_2)$).
Recall that if $Dir(G_1)$ is a symmetric covering
of the symmetric digraph $D$, then the number of vertices of $D$ must divide that of $Dir(G_1)$  (see, \textit{e.g.} Proposition 5 in~\cite{BVfibrations}), thus it must be either $1$, $2$, $4$, $7$, or $14$. By inspection, we show that $Dir(G_1)$ is minimal in each case.

\begin{itemize}
\item If $|V(D)| = 1$, then $D$ is a $3$-bouquet, and necessary of one the
  loop is its own symmetric. As a result, the inverse image of this
  loop is a perfect matching in $Dir(G_1)$, which is a contradiction because $G_1$ ($Dir(G_1)$) admits no perfect matching.
\item If $|V(D)| = 2$, then once again the inverse image
  of an arc between the two vertices (and its symmetric in $D$) must form a perfect matching of
  $Dir(G_1)$, which gives the same contradiction.
\item If $|V(D)| = 4$, then $Dir(G_1)$ can be viewed
  as a symmetric covering of $D$ through the construction of
  Reidemeister applied to a spanning tree of $D$.  Let $v_0$, $v_1$, $v_2$, and
  $v_3$ be the vertices of $D$ and consider any spanning tree over these vertices. The possible graphs defined by four vertices in
  the subgraph $C_1$ imply that at least one of the images, by the
  covering mapping, of the edges $\{v_0,v_1\}$, $\{v_0,v_2\}$,
  $\{v_0,v_3\}$ is not a bridge in $D$, thus one of the
  edges $\{v_0,v_1\}$, $\{v_0,v_2\}$, $\{v_0,v_3\}$ is not a bridge
  in $G_1$ (contradiction).
\item If $|V(D)| = 7$, then it is impossible to cover $Dir(G_1)$ with four (connected) trees of size~$7$.
\item If $|V(D)| = 14$, then it is impossible to cover $Dir(G_1)$ with two trees of equal sizes.
  
\end{itemize}

\section{Conclusion}

In this paper, we characterized necessary and sufficient conditions for a graph to admit a distributed algorithm for the spanning tree construction and topology recognition problems. The conditions are expressed in terms of graph coverings, which fills a long-standing gap in the distributed computing literature. While carrying ideas similar to those of Yamashita and Kameda's proofs based on the concepts of views and quotient graphs, our proofs invoke general tools in covering theory (Reidemeister's theorem, Mazurviewicz's algorithm, and the lifting lemma) that make the arguments simpler and more natural once these tools are known. An important question that remains open is whether the condition for topology recognition can ultimately be expressed in terms of graph properties which are more elementary than coverings. Also, whether the corresponding property could be tested in polynomial time on a given graph remains an open question, arguably related to the first question.

\bibliographystyle{alpha}
\bibliography{article}

\newcommand{\etalchar}[1]{$^{#1}$}
\begin{thebibliography}{CGMT07}

\bibitem[AAD{\etalchar{+}}06]{population}
D.~Angluin, J.~Aspnes, Z.~Diamadi, M.J. Fischer, and R.~Peralta.
\newblock Computation in networks of passively mobile finite-state sensors.
\newblock {\em Distributed computing}, 18(4):235--253, 2006.

\bibitem[AG81]{Gardiner}
D.~Angluin and A.~Gardiner.
\newblock Finite common coverings of pairs of regular graphs.
\newblock {\em J. Combin. Theory, Ser. B}, 30:183--187, 1981.

\bibitem[ALMR01]{random-lifts}
A.~Amit, N.~Linial, J.~Matou{\v{s}}ek, and E.~Rozenman.
\newblock Random lifts of graphs.
\newblock In {\em Proceedings of the twelfth annual ACM-SIAM symposium on
  Discrete algorithms}, pages 883--894, 2001.

\bibitem[Ang80]{Angluin}
D.~Angluin.
\newblock Local and global properties in networks of processors.
\newblock In {\em Proceedings of the {12th} Symposium on Theory of Computing},
  pages 82--93, 1980.

\bibitem[BCG{\etalchar{+}}96]{BVelection}
P.~Boldi, B.~Codenotti, P.~Gemmell, S.~Shammah, J.~Simon, and S.~Vigna.
\newblock Symmetry breaking in anonymous networks: Characterizations.
\newblock In {\em Proc. 4th Israeli Symposium on Theory of Computing and
  Systems}, pages 16--26. IEEE Press, 1996.

\bibitem[Bod89]{Bod}
H.-L. Bodlaender.
\newblock The classification of coverings of processor networks.
\newblock {\em Journal of parallel and distributed computing}, 6:166--182,
  1989.

\bibitem[BV01]{BVanonymous}
P.~Boldi and S.~Vigna.
\newblock An effective characterization of computability in anonymous networks.
\newblock In Jennifer~L. Welch, editor, {\em Distributed Computing. 15th
  International Conference, DISC 2001}, volume 2180 of {\em Lecture Notes in
  Computer Science}, pages 33--47. Springer-Verlag, 2001.

\bibitem[BV02]{BVfibrations}
P.~Boldi and S.~Vigna.
\newblock Fibrations of graphs.
\newblock {\em Discrete Math.}, 243:21--66, 2002.

\bibitem[CFSV16]{CFSV16}
K.~Censor{-}Hillel, E.~Fischer, G.~Schwartzman, and Y.~Vasudev.
\newblock Fast distributed algorithms for testing graph properties.
\newblock In {\em Distributed Computing - 30th International Symposium, {DISC}
  2016, Paris, France, September 27-29, 2016. Proceedings}, pages 43--56, 2016.

\bibitem[CGM12]{CGM12}
J.~Chalopin, E.~Godard, and Y.~M{\'{e}}tivier.
\newblock Election in partially anonymous networks with arbitrary knowledge in
  message passing systems.
\newblock {\em Distributed Computing}, 25(4):297--311, 2012.

\bibitem[CGMT07]{CGMT07}
J.~Chalopin, E.~Godard, Y.~M{\'{e}}tivier, and G.~Tel.
\newblock About the termination detection in the asynchronous message passing
  model.
\newblock In {\em {SOFSEM} 2007: Theory and Practice of Computer Science, 33rd
  Conference on Current Trends in Theory and Practice of Computer Science,
  Harrachov, Czech Republic, January 20-26, 2007, Proceedings}, pages 200--211,
  2007.

\bibitem[Cha06]{Cthese}
J.~Chalopin.
\newblock {\em Algorithmique distribu{\'e}e, calculs locaux et homorphismes de
  graphes}.
\newblock PhD thesis, universit{\'e} Bordeaux 1, 2006.

\bibitem[CM94]{CMminors}
B.~Courcelle and Y.~M{\'e}tivier.
\newblock Coverings and minors: Applications to local computations in graphs.
\newblock {\em Europ. Journal of Combinatorics}, 15:127--138, 1994.

\bibitem[CM07]{CMasynj}
J.~Chalopin and Y.~M{\'e}tivier.
\newblock An efficient message passing election algorithm based on
  {M}azurkiewicz's algorithm.
\newblock {\em Fundam. Inform.}, 80(1-3):221--246, 2007.

\bibitem[CP11]{paulusma}
J.~Chalopin and D.~Paulusma.
\newblock Graph labelings derived from models in distributed computing: {A}
  complete complexity classification.
\newblock {\em Networks}, 58(3):207--231, 2011.

\bibitem[DP16]{DP16}
D.~Dereniowski and A.~Pelc.
\newblock Topology recognition and leader election in colored networks.
\newblock {\em Theor. Comput. Sci.}, 621:92--102, 2016.

\bibitem[FLP85]{FLP85}
M.J. Fischer, N.A. Lynch, and M.S. Paterson.
\newblock Impossibility of distributed consensus with one faulty process.
\newblock {\em Journal of the ACM (JACM)}, 32(2):374--382, 1985.

\bibitem[FRST16]{FRST16}
P.~Fraigniaud, I.~Rapaport, V.~Salo, and I.~Todinca.
\newblock Distributed testing of excluded subgraphs.
\newblock In {\em Distributed Computing - 30th International Symposium, {DISC}
  2016, Paris, France, September 27-29, 2016. Proceedings}, pages 342--356,
  2016.

\bibitem[GMM04]{GMMrecog}
E.~Godard, Y.~M\'etivier, and A.~Muscholl.
\newblock Characterization of {C}lasses of {G}raphs {R}ecognizable by {L}ocal
  {C}omputations.
\newblock {\em Theory of Computing Systems}, (37):249--293, 2004.

\bibitem[GT87]{GT}
J.~L. Gross and Th.~W. Tucker, editors.
\newblock {\em Topological graph theory}.
\newblock Wiley-Interscience, 1987.

\bibitem[KKP16]{KKP16}
R.~Klasing, A.~Kosowski, and D.~Pajak.
\newblock Setting ports in an anonymous network: How to reduce the level of
  symmetry?
\newblock In {\em International Colloquium on Structural Information and
  Communication Complexity}, pages 35--48. Springer, 2016.

\bibitem[Lei82]{L82}
Frank~Thomson Leighton.
\newblock Finite common coverings of graphs.
\newblock {\em Journal of Combinatorial Theory, Series B}, 33(3):231--238,
  1982.

\bibitem[Lyn96]{Lynch}
N.~A. Lynch.
\newblock {\em Distributed algorithms}.
\newblock Morgan Kaufman, 1996.

\bibitem[Maz97]{MazurEnum}
A.~Mazurkiewicz.
\newblock Distributed enumeration.
\newblock {\em Information Processing Letters}, 61(5):233--239, 1997.

\bibitem[Nor95]{Norris}
N.~Norris.
\newblock Universal covers of graphs: isomorphism to depth $n-1$ implies
  isomorphism to all depths.
\newblock {\em Discrete Applied Math.}, 56:61--74, 1995.

\bibitem[PR02]{PR02}
Ch. Papazian and E.~R{\'{e}}mila.
\newblock Hyperbolic recognition by graph automata.
\newblock In {\em Automata, Languages and Programming, 29th International
  Colloquium, {(ICALP)}}, pages 330--342, 2002.

\bibitem[Rei32]{Reidemeister}
K.~Reidemeister.
\newblock {\em Einf\"uhrung in die Kombinatorische Topologie}.
\newblock Vieweg, Brunswick, 1932.

\bibitem[R{\'{e}}m94]{R94}
E.~R{\'{e}}mila.
\newblock Recognition of graphs by automata.
\newblock {\em Theor. Comput. Sci.}, 136(2):291--332, 1994.

\bibitem[RFH72]{Rosen}
P.~Rosenstiehl, J.-R. Fiksel, and A.~Holliger.
\newblock Intelligent graphs.
\newblock In R.~Read, editor, {\em Graph theory and computing}, pages 219--265.
  Academic Press (New York), 1972.

\bibitem[Ros75]{R75}
A.~Rosenfeld.
\newblock Networks of automata: Some applications.
\newblock {\em {IEEE} Trans. Systems, Man, and Cybernetics}, 5(3):380--383,
  1975.

\bibitem[Tan11]{compressed}
S.~Tani.
\newblock Compression of view on anonymous networks—folded view—.
\newblock {\em IEEE Transactions on Parallel and Distributed Systems},
  23(2):255--262, 2011.

\bibitem[Tar95]{Tarry}
G.~Tarry.
\newblock Le probl\`eme des labyrinthes.
\newblock {\em Nouvelles Annales de Math\'ematiques}, 3(14):187--190, 1895.

\bibitem[Tel00]{Tel}
G.~Tel.
\newblock {\em Introduction to distributed algorithms}.
\newblock Cambridge University Press, 2000.

\bibitem[WR79]{WR79}
A.~Y. Wu and A.~Rosenfeld.
\newblock Cellular graph automata. {II.} graph and subgraph isomorphism, graph
  structure recognition.
\newblock {\em Information and Control}, 42(3):330--353, 1979.

\bibitem[YK88]{YK88}
M.~Yamashita and T.~Kameda.
\newblock Computing on anonymous networks.
\newblock In {\em Proc. 7th ACM Symposium on Principles of Distributed
  Computing (PODC)}, pages 117--130, 1988.

\bibitem[YK96]{YKsolvable}
M.~Yamashita and T.~Kameda.
\newblock Computing on anonymous networks: Part i - characterizing the solvable
  cases.
\newblock {\em IEEE Transactions on parallel and distributed systems},
  7(1):69--89, 1996.

\bibitem[YK99]{YKelection}
M.~Yamashita and T.~Kameda.
\newblock Leader election problem on networks in which processor identity
  numbers are not distinct.
\newblock {\em IEEE Transactions on parallel and distributed systems},
  10(9):878--887, 1999.

\end{thebibliography}

\newpage
\appendix
\section{Constructing a spanning tree is equivalent to electing a leader or a pair of neighbor co-leaders}
\label{sec:coleader}

This appendix exposes the well known fact that the spanning tree construction problem and the edge election problem are equivalent in terms of feasibility. The high-level arguments are recalled in Section~\ref{sec:problems}. We give here the detailed algorithms used in the equivalence.

\begin{definition}
Let {\tt leader} and {\tt co-leader} be two states. 
A vertex is said to be distinguished (or simply to be the leader) if
it is the only vertex in the {\tt leader} state. A pair of
neighbor vertices is said to be distinguished (or to be co-leader) if their states
are {\tt co-leader}, and no other vertex is in this state.
\end{definition}

Algorithm~\ref{algo:coleader} below adapts ideas from Angluin (\cite{Angluin}, Theorem 4.4) to the asynchronous message passing model. Similar arguments are also developed in Yamashita and Kameda~\cite{YKsolvable}. A tree $T$ of size $n>1$ is initially
given.
Informally, the leader  vertex (resp.  the two co-leader vertices) will be the last
 surviving vertex (resp. the last two to be eliminated), following an
$(n-1)$-sequence (resp. a $n$-sequence) of leaf-eliminations. 
 We use the same notation conventions as in Tel
(\cite{Tel} page 192, Algorithm 6.3).

\noindent
\begin{algorithm}
  \label{algo:coleader}
  \caption{Election of a $leader$ or a pair of $co$-$leaders$ in a tree.}
{\bf var }    $Neigh_v$ : set of neighbors of $v$;\\
\hspace*{0.8cm}$rec_v[w]$ for each $w\in Neigh_v$ : boolean {\bf init} false;\\
\hspace*{0.8cm}(* $rec_v[w]$ is true if $v$ has received a message
from $w$ *);\\
{\bf begin}\\
\hspace*{0.5cm}{\bf while} $|\{w: rec_v[w]$ is false\}$|>1$ {\bf do}\\
\hspace*{0.8cm}{\bf begin}\\
\hspace*{1cm}receive $<${\bf tok}$>$ from $w$;\\
\hspace*{1cm}$rec_v[w]:=true$\\
\hspace*{0.8cm}{\bf end}\\
\hspace*{0.5cm}({\small * Now there is exactly 
$one$ $w_0$ with $rec_v[w_0]$ false: the
           vertex $v$ is a leaf *})\\
\hspace*{0.5cm}{\bf if} {a message $<${\bf tok}$>$ has arrived}
	{\bf then}
$state_v:=leader$\\
\hspace*{0.5cm}{\bf else}\\
\hspace*{0.8cm}{\bf begin}\\
\hspace*{1cm}send $<${\bf tok}$>$ to $w_0$ with $rec_v[w_0]$ is false;\\
\hspace*{1.cm}{\bf if} {a message $<${\bf tok}$>$ has arrived}
	{\bf then}
$state_v:=co$-$leader$\\
\hspace*{0.8cm}{\bf end}\\
{\bf end}
\end{algorithm}

\begin{center}
\end{center}

The algorithm is initiated by leaves.
Each vertex waits for messages until it
has received messages from all but one of its neighbors.
When a vertex $v$ has received a message from all but one neighbor,
denoted $v'$,
then it sends a message to $v'$.
A leaf $v$ which has  not received
a message 
 sends a message to its unique neighbor.
 A vertex which has received messages from all but one 
 of its neighbors, denoted $v'$,
 becomes $leader$
 if it receives a message from $v'$ before it, itself, sends a message
 to $v'$.
Hence, the $leader$  is the unique vertex which
received a message from all neighbors without sending
a message (see \cite{Tel} page 193).
The two $co$-$leader$ vertices are the two vertices which receive
a message after they have themselves sent a message.
It follows from this algorithm and Theorem 6.16 \cite{Tel} that:
\begin{proposition}
  Let $T$ be a tree. There exists a distributed algorithm
  which distinguishes  a vertex
or a pair of neighboring vertices.
\end{proposition}
Thus, given a graph $G$ and a spanning tree $T$ of $G$, there 
exists a distributed algorithm which distinguishes a  vertex
or a pair of neighboring vertices of $G$.

Conversely, let $G$ be a graph with one distinguished vertex or 
two neighboring distinguished vertices. There exists a distributed algorithm
which computes a spanning tree of $G$. For example, Tarry's algorithm 
\cite{Tarry}
presented by Tel (\cite{Tel}, p. 207) computes a spanning tree. It is
formulated in the following two rules:
\begin{enumerate}
\item a process never forwards the token twice through the same channel;
\item a non-$leader$ (or non-$co$-$leader$)
  forwards the token to its parent (neighbor from which it first received it) only if there is no other
channel possible according to the first rule.
\end{enumerate}
The $leader$ vertex (resp. the two $co$-$leader$ vertices) is initiator 
(resp. are initiators) sending the token
the first time to a neighbor (different from a $co$-$leader$ vertex if any).
Algorithm 1 presents a distributed implementation of Tarry's algorithm.

\begin{algorithm}
{\scriptsize
  ${\mathbf I:}$ \{The vertex $v$ $leader$ or  $co$-$leader$ owns $<$\KwSty{token}$>$ and
  initializes the spanning tree computation\}\\ 
   \Begin{\KwSty{send}$<$\KwSty{token}$>$ 
via each  port (different from $co$-$leader$ if any)
}
\BlankLine
${\mathbf R_1:}$ \{Upon receipt $<$\KwSty{token}$>$  by $v$ via the port $q$\}\\
   \Begin{
       \eIf{$v$ owns $<$\KwSty{token}$>$   for the first time}
       {$Tree_v:=Tree_v \cup \{q\}$;\\
         \KwSty{send}$<$\KwSty{in-the-tree}$>$ via the port $q$\\
         \KwSty{send}$<$\KwSty{token}$>$ via each port of $v$ different from $q$
        }
{\KwSty{send}$<$\KwSty{already-in-the-tree}$>$ via the port $q$\\
$other_v:=other_v \cup \{q\}$
}
}
${\mathbf R_2:}$ \{Upon receipt $<$\KwSty{already-in-the-tree}$>$
 by $v$ via the port $q$\}\\
{\Begin{$other_v:=other_v \cup \{q\}$}
}
${\mathbf R_2:}$ \{Upon receipt $<$\KwSty{in-the-tree}$>$
 by $v$ via the port $q$\}\\
{\Begin{$Tree_v:=Tree_v \cup \{q\}$}
}
}
\caption{Spanning tree computation: for each vertex $v$, the variable $Tree_v$
  indicates  which ports of $v$ are in the spanning tree. Its initial value is the empty set
  if there is one $leader$ vertex; if there are two $co$-$leader$ vertices then it is the empty set
  for vertices which are not $co$-$leader$, and for each $co$-$leader$ vertex it contains the port number
  corresponding to its $co$-$leader$ neighbor. The variable
  $other_v$ is the set of ports which are not in the 
spanning tree (its initial value is the empty set).}
\end{algorithm}

\begin{remark}
  Algorithm 1 has an implicit termination. To obtain an explicit termination:
  each vertex remembers its parent (the port number through wich it receives the token the first time)
  and the set of its children (port numbers through which it receives the message $in$-$the$-$tree$).
Then an acknowledgment is initialized by leaves,
and relayed to its parent by a vertex as soon as it has received an acknowledgement from
each child; it  enables the $leader$ vertex  to detect the termination
in the case there is a $leader$ vertex.
If there are two $co$-$leader$ vertices then each $co$-$leader$ vertex sends an acknowledgement
to its neighbor which is a $co$-$leader$ vertex  as soon as it has received an
acknowledgement from each child.
\end{remark}

\section{Message-passing adaptation of Mazurkiewicz's algorithm}\label{algo}
\label{sec:mazurkiewicz-full}

This section describes, for compleness, an algorithm introduced in~\cite{CMasynj}, denoted $\mk$,
which, given a graph $G$ and a port numbering $\delta$ of $G$,  
can be seen as computing a symmetric digraph, denoted $\rho(G,\delta)$,
covered by $(Dir(G),\delta).$ Furthermore $\mk$ enumerates vertices
of $\rho(G,\delta)$.
Conversely, 
every symmetric digraph $D'$ covered by $(Dir(G),\delta)$ and every enumeration
of vertices of $D'$ 
can be obtained by a run $\rho$ of $\mk$. This algorithm is an adaptation to the asynchronous message passing setting, of the coarse-grain algorithm introduced by Mazurkiewicz~\cite{MazurEnum} in a model based on atomic transformations of local stars.

\subsection{A General Description of $\mk$}

We  give a general description of  $\mk$, 
when executed on a connected  simple graph
$G$ with a port numbering $\delta$.
During the execution of the algorithm, each vertex $v$ attempts to get
its own unique identity which is a number between $1$ and
$|V(G)|$. Once a vertex $u$ has chosen a number $n(u)$, it sends it to
each neighbor $v$ with the port number $\delta_u(v)$. When a vertex $u$
receives a message from one neighbor $v$, it stores the number $n(v)$
with the port numbers $\delta_u(v)$ and $\delta_v(u)$. From all 
information it has gathered from its neighbors, each vertex can
construct its \emph{local view} (which is the set of numbers of its
neighbors associated with the corresponding port numbers). Then, a
vertex broadcasts its number with its {\em local view}.  If a vertex
$u$ discovers the existence of another vertex $v$ with the same number
then it should decide if it changes its identity. To this end it
compares its local view with the local view of $v$. If
 the local view of $u$ is ``weaker'' (the order is defined below), then $u$ picks another
number --- its new temporary identity --- and broadcasts it again with
its local view. 
\begin{remark}
At the end of the computation of $\mk$ on a graph $G$, each vertex has a number and
can, thanks to some local information (the set of local views of the vertices of $G$), 
build a graph $\rho(G,\delta)$ such that
$(Dir(G),\delta)$ is a symmetric covering of $\rho(G,\delta)$. 

Thus, in the case where $Dir(G)$ is minimal
for the symmetric covering relation, each vertex can build $G$, 
and every vertex of $G$ 
will have a unique number between $1$ and the size of $G$:
the algorithm is an enumeration (naming) algorithm.

In the case where $(Dir(G),\delta)$ is not minimal and $\rho(G,\delta)$ is not, in general, isomorphic to 
$(Dir(G),\delta)$,
then $(Dir(G),\delta)$ is a symmetric
covering of $\rho(G,\delta)$, each vertex can build $\rho(G,\delta)$ and each vertex
of $\rho(G,\delta)$ (and of $G$) 
 will have a unique number between $1$ and the size of $\rho(G,\delta)$.
\end{remark}
\subsection{Labels}

We consider a network $(G,\delta)$ where $G$ is a
simple  graph and where $\delta$ is a port numbering of
$G$.  During the
execution, the label of each vertex $v$ is a tuple $(n(v),
N(v), M(v))$ representing  following information:
\begin{itemize}
\item $n(v) \in \N$ is the current {\em number} of the vertex $v$
computed by the algorithm;
\item $N(v) \in \NNNN$\footnote{For any set $S$,
${\mathcal{P}}(S)$ denotes the power set of
$S.$} is the {\em local view} of $v$. The local view of a vertex $v$
contains the information a vertex $v$ has about its neighbors. If a
vertex $v$ has a neighbor $u$ such that $\delta_u(v) = p$ and $\delta_v(u)
= q$, then $(m,p,q) \in N(v)$ if the last message that $v$ got from
$u$ indicated that $n(u) = m$;
\item $M(v) \in \N  \times \NNNN$ is the {\em mailbox} of $v$. The
mailbox of $v$ contains all information received by $v$ during the
execution of the algorithm. If $(m,{\cal N}) \in M(v)$, it means that
at some previous step of the execution, there was a vertex $u$ such
that $n(u) = m$,  and $N(u) = {\cal N}$. 
\end{itemize}

Initially, each vertex $v$ has a label of the form
$(0,\emptyset,\emptyset)$ indicating that it has not
choosen any number, that it has no information about its
neighbors or about the other vertices of the graph. 

During an excecution of $\mk$, processes exchange messages of the form  $ < (m,n_{old},M),p >$. 
If a vertex $u$ sends a message $ < (m,
n_{old},M),p >$ to one of its neighbor $v$, then the message contains
the following information:
\begin{itemize}
\item $m$ is the current number $n(u)$ of $u$;
\item $n_{old}$ is the previous number of $u$, i.e., the number $u$
sends to $v$ in its previous message; if in the meanwhile, $u$ has
not modified its number, then $n_{old} = m$;
\item $M$ is the mailbox of $u$;
\item $p$ is the port number the message has been sent through, i.e.,
$ p =\delta_{u}(v)$.
\end{itemize}

\subsection{An Order on Local Views}

As in Mazurkiewicz's algorithm~\cite{MazurEnum}, the nice properties of
the algorithm rely on a total order on local views, i.e., on finite
subsets of $\N^3$. We consider the usual lexicographic order on
$\N^3$: $(n,p,q) < (n',p',q')$ if $n < n'$, or if $n = n'$ and $p <p'$,
or if  $n = n'$, $p = p'$ and $q < q'$.
Then, we use the same order on finite sets as Mazurkiewicz: given two
distinct sets $N_1, N_2 \in \NNNN$, we define $N_1 \prec N_2$ if the
maximum of the symmetric difference $N_1 \bigtriangleup N_2 = (N_1
\setminus N_2) \cup (N_2 \setminus N_1)$ belongs to $N_2$.
If $N(u) \prec N(v)$, then we say that the local view $N(v)$ of $v$ is
\emph{stronger} than the local view $N(u)$ of $u$ and that $N(u)$ is
\emph{weaker} than $N(v)$. 
We denote by $\preceq$ the reflexive closure
of $\prec$.

\subsection{Algorithm $\mk$}

The algorithm for the vertex $v_0$ (see Algorithm \ref{maz}) 
is expressed in an event-driven description (see Tel 
\cite{Tel} p. 553). 
The action ${\mathbf I}$ can be executed by a process on wake-up,
only if it has not received any message. In this case, it chooses the
number $1$, updates its mailbox and informs its neighbors.
The action ${\mathbf R}$ describes the instructions the vertex
$v_0$ has to follow when it receives a message $< (n',n_{old}',M'), p
> $ from a neighbor via port $q$. First, it updates its mailbox by
adding $M'$ to it. Then it modifies its number if there exists
$(n(v_0),{\cal N}) \in M(v_0)$ such that $N(v_0) \prec
{\cal N}$. 
Then, it updates its local view by removing
$(n_{old}',p,q)$ from $N(v_0)$ (if $N(v_0)$ contains such an element)
and by adding $(n',p,q)$ to $N(v_0)$. Then, it adds its new state
$(n(v_0), N(v_0))$ to its mailbox. Finally, if its
mailbox has been modified by the execution of all these instructions,
it sends its number and its mailbox to all its neighbors. 

If the mailbox of $v_0$ is not modified by the execution of the action
${\mathbf R}$, it means that the information $v_0$ has about its neighbor
(i.e., its number) was correct, that all the elements of $M'$ already
belong to $M(v_0),$ and that for each $(n(v_0),{\cal N}) \in M(v_0)$, ${\cal N}
\preceq N(v_0)$.

\begin{algorithm}
${\mathbf I:}$ \{$n(v_0)=0$ and no message has arrived at $v_0$\}\\
\Begin{ $n(v_0):=1$ \; $M(v_0):=\{(n(v_0),\emptyset)\}$\;
\For{$i:=1$ \KwTo $\deg(v_0)$} {\KwSty{send} $<(n(v_0),0,M(v_0)),i>$
via port $i$ \;} }
\BlankLine
${\mathbf R:}$ \{A message $<(n',n_{old}',M'),p>$ has arrived at $v_0$ from port $q$\}\\
   \Begin{
      $M_{old} := M(v_0)$\;
      $n_{old} := n(v_0)$\;
      $M(v_0):= M(v_0)\cup M'$\;
      \If{$n(v_0) = 0$ or $\exists (n(v_0),{\cal N}) \in M(v_0) \mbox{
	  such that } N(v_0) \prec {\cal N}$} 
          {$n(v_0):=1 + \max \{ n \mid  \exists (n,{\cal N}) \in M(v_0)\}$\;}
      $N(v_0):= N(v_0) \setminus \{(n_{old}',p,q)\} \cup \{(n',p,q) \}$\;
      $M(v_0) := M(v_0)\cup \{(n(v_0),N(v_0))\}$\;
      \If{$M(v_0)\neq M_{old}$}
        {\For{$i:=1$ \KwTo $\deg(v_0)$}
	   {\KwSty{send} $<(n(v_0),n_{old},M(v_0)),i>$ through port $i$\;}}   
}
\caption{Algorithm $\mk$.}
\label{maz}
\end{algorithm}

As is explained in \cite{CMasynj}, an execution of  $\mk$ on a graph $G$ always terminates.
The algorithm $\mk$ does not need any information on $G$.
However, the vertices cannot detect that the computation is over 
if they have no
information on the graph: it is an implicit termination.
A global termination detection
can be obtained under the following hypotheses \cite{CGMT07}: the knowedge of the 
size of $G$, or the knowledge of an upper bound on the size of $G$.

\begin{lemma}\label{lem-fundamental}
Any execution $\rho$ of $\mk$ on a simple labeled graph $G$
 with port numbering $\delta$ terminates and the final
labeling $(n_\rho, N_\rho, M_\rho)$ 
satisfies the following properties:
 \begin{enumerate}
   \item there exists an integer  $k \leq |V(G)|$ such that
   $\{n_\rho(v) \mid v \in V(G)\} = [1,k]$, 
 \end{enumerate}
  and for all vertices $v, v' \in V(G)$:
  \begin{enumerate}
   \setcounter{enumi}{1}
   \item $M_\rho(v)=M_\rho(v')$,
   \item $(n_\rho(v),N_\rho(v)) \in M_\rho(v')$,
   \item if $n_\rho(v) = n_\rho(v')$, then   $ N_\rho(v) = N_\rho(v')$, 
   \item $(n,p,q) \in N_\rho(v)$ if and only if there exists
   $w \in N_G(v)$ such that $\delta_v(w) = q$, $\delta_w(v) = p$ and
   $n_\rho(w) = n$.
\end{enumerate}
\end{lemma}
\begin{proposition}\label{prop-graphe-final}
Given a graph $G$ with a port numbering $\delta$, we can associate
with the final labeling of any execution $\rho$ of $\mk$ on
$(G,\delta)$, a labeled
digraph $(D,\delta')$ 
such that $(Dir(G),\delta)$ is a
symmetric covering of $(D,\delta')$. Furthermore, each vertex of $G$ can compute
$(D,\delta')$, and vertices of $D$ are enumerated by numbers of $[1,|D|]$.
\end{proposition}

The proofs of the above properties can be found in \cite{CMasynj,CGMT07,CGM12}.

\end{document}